\documentclass[12pt]{article}

\usepackage{amsmath}
\allowdisplaybreaks 

\usepackage{natbib}
\usepackage[hidelinks, colorlinks=true, linkcolor=black, citecolor=blue]{hyperref}
\usepackage{amsthm}
\usepackage{subcaption}
\usepackage{multirow}

\newtheorem{theorem}{Theorem}[section]

\newtheorem{corollary}[theorem]{Corollary}

\newcommand{\ssection}[1]{%
  \section[#1]{\centering\scshape #1}}

\newcommand{\be}{\begin{eqnarray}}
\newcommand{\ee}{\end{eqnarray}}
\newcommand{\ben}{\begin{eqnarray*}}
\newcommand{\een}{\end{eqnarray*}}

\usepackage{bm}
\usepackage{mathrsfs}

\def\NONUMBIB{\def\thebibliography##1{\addvspace{12pt plus 4pt
  minus 2pt}\begin{center}\Large\bf REFERENCES\end{center}
  \addvspace{6pt}\footnotesize
  \list{}{\labelwidth0pt
    \leftmargin40pt
    \itemindent-40pt
    \itemsep=0pt\parsep=0pt
    \usecounter{enumiv}%
    \def\theenumiv{\arabic{enumiv}}}%
    \def\newblock{\hskip .11em plus.33em minus.07em}%
    \sloppy\clubpenalty4000\widowpenalty4000
    \sfcode`\.=1000\relax}}
\NONUMBIB
\usepackage{amsmath}
\usepackage{epsfig}
\usepackage{graphics}
\usepackage{amssymb}
\usepackage{color}  
\usepackage{verbatim}
\begin{document}

\begin{flushright}
{\large{\bf Pairwise Sequential Randomization and Its Properties}}
\end{flushright}

\bigskip
\bigskip

\begin{flushleft}
{\sc By Yichen Qin, Yang Li, Wei Ma, and Feifang Hu}\footnote{Yichen Qin is 
Assistant 
Professor of Business Analytics, University of Cincinnati.  Yang Li is Associate Professor of Statistics, Renmin 
University of China.  Wei Ma is Assistant 
Professor of Statistics, Renmin University of China.  Feifang Hu (e-mail: 
feifang@gwu.edu) is Professor of Statistics, Department of Statistics, George Washington 
University, Washington, DC 20052.  The research was partially supported by NSF 
Awards (DMS-1612970) and the National Natural Science 
Foundation of China (No. 11371366 and No. 71301162).}
\end{flushleft}

\newpage

\begin{flushright}
{\large{\bf Pairwise Sequential Randomization and Its Properties}}
\end{flushright}

\bigskip
\bigskip

\bigskip
\bigskip
\begin{quote}
{\em SUMMARY:} In comparative studies, such as in causal 
inference and clinical trials,  balancing important covariates is often one of the 
most important concerns for both efficient and credible comparison.   
However, chance imbalance still exists in many randomized experiments. 
This phenomenon of covariate imbalance becomes much more serious as the number 
of covariates $p$ increases. 
To address this issue,  we introduce a new randomization procedure, 
called pairwise sequential randomization (PSR). The proposed method allocates the units sequentially and 
adaptively, using information on the current level of imbalance and the 
incoming unit's covariate.  
With a large number of covariates or a large number of units, 
the proposed method shows  
substantial advantages over the traditional methods in terms of the covariate 
balance, estimation accuracy, and computational time, making it an ideal 
technique in the era of big data.  
The proposed method attains the optimal covariate balance, 
in the sense that the estimated treatment effect under the proposed 
method attains its minimum variance asymptotically.  Also the proposed method is widely applicable in both 
causal inference and clinical trials. Numerical studies and real 
data analysis provide further evidence of the advantages of the proposed method.

\end{quote}
\newpage
\pagenumbering{arabic}

\section{Introduction}\label{sec:intro}

Randomization is the foundation for the treatment effect evaluation.
However, traditional randomization methods often generate unsatisfactory configurations with unbalanced prognostic covariates; 
this issue has been extensively discussed ever since \citet{Fisher1926} noted: 
``{\it Most of experimenters on carrying out a random 
assignment of plots will be shocked to find out how far from equally the plots 
distribute themselves.}''
The advantages of balanced covariates are at least 
threefold \citep{Hu2014}.  First, covariate balance improves the efficiency of 
estimation for the treatment effect.
Second, it increases the interpretability of the estimated treatment effect by 
making the units in the treatment groups more comparable, thereby enhancing the 
credibility of the analysis.  Third, it makes the analysis more robust against 
model misspecification. 
Consequently, covariate imbalance can 
significantly undermine the validity of subsequent analysis.  In the absence of covariate balance, various problems must be addressed before 
a valid conclusion can be drawn.  

In causal inference and clinical studies, 
if a significant imbalance exists, any inferences regarding the treatment effect will be inaccurate, 
and any claims about the treatment effect will need 
to rely on unverifiable assumptions \citep{Morgan2011}.  
Researchers must assess the balance in the covariate distribution before 
estimating the causal effect.  
Although some ex-post adjustments, such as regression \citep{Freedman2008} and subsample selection 
using matching or trimming based on propensity scores \citep{Imbens2015}, 
can cope with such an imbalance, 
they are much less efficient than achieving an ex-ante balance from the start \citep{Bruhn2008}.  
In addition, these adjustments often rely on at 
least a nearly correct model, which can be difficult to test \citep{Cochran1965,Cochran1973}.  \citet{Rubin2008b} explained why the greatest 
possible efforts should be made during the design phase of an experiment rather 
than during the analysis stage, at which point the researcher has the potential 
to bias the results \citep{Morgan2011,Imbens2015}.  


More recently, covariate balance has attracted growing interest in the field of crowdsourced-internet experimentation \citep{Horton2011,Chandler2013,Kapelner2014}.  Researchers increasingly recruit workers from online labor markets into their experiments, such as by asking them to label tumor cells in images.  Because of the nature of the recruiting process, a large number of workers with many covariates (e.g., 2500 workers in \citet{Chandler2013}), typically are enrolled in such studies, which consequently pose challenges for traditional randomization methods.  

Furthermore, the phenomenon of covariate imbalance is exacerbated as the number 
of covariates $p$ and the sample size $n$ increase, which is nearly ubiquitous 
in the era of big data.  For example, suppose that the probability of one 
particular covariate being unbalanced is $5\%$.  
For a study with 10 independent covariates, 
the chance of at least one covariate exhibiting imbalance is 
$1-(1-5\%)^{10}=40\%$.  
Meanwhile, some may argue that imbalance tends to be milder as the sample size 
increases. 
However, as the sample size increases, 
even though the difference 
in covariate means between groups becomes smaller, however, at the same rate, confidence intervals and hypothesis testing are becoming
more sensitive to small differences in outcome variables which can be affected 
by the small imbalance in covariates \citep{Morgan2012}.

In the framework of causal inference, \citet{Morgan2012} have proposed rerandomization 
(RR).
They propose to repeatedly randomize the units into treatment groups using complete randomization (CR), until certain the balance criterion is satisfied, e.g., $M<a$, where $M$ is the Mahalanobis 
distance between the sample means across different treatment groups and $a>0$ 
is a threshold.
\begin{align*}
M&=(\boldsymbol{\bar{x}}_1-\boldsymbol{\bar{x}}_2)^T
[\textrm{cov}(\boldsymbol{\bar{x}}_1-\boldsymbol{\bar{x}}_2)]^{-1}
(\boldsymbol{\bar{x}}_1-\boldsymbol{\bar{x}}_2)\\
&\propto (\boldsymbol{\bar{x}}_1-\boldsymbol{\bar{x}}_2)^T
\textrm{cov}(\boldsymbol{x})^{-1}
(\boldsymbol{\bar{x}}_1-\boldsymbol{\bar{x}}_2),
\end{align*}
where $\boldsymbol{\bar{x}}_1 \in \mathbb{R}^p$ and $\boldsymbol{\bar{x}}_2 \in 
\mathbb{R}^p$ are the sample 
means for two treatment groups, $\textrm{cov}(\boldsymbol{x}) \in \mathbb{R}^{p 
\times p}$ is the covariance matrix of the covariate. 
They has also assumed fixed equal numbers of units in two treatment groups and demonstrated various desirable properties under rerandomization.

Although rerandomization works well in the case of a few covariates, it is 
incapable of scaling up to address massive amounts of data.  For example, as 
the number of covariates increases, the probability of acceptance, 
$p_a=P(M<a)$, of each complete randomization decreases drastically, 
causing the rerandomization procedure to remain in loop for a long time.  
To compromise the computational burden, one can increase $a$, which unavoidably leads poorer covariate imbalance.

In clinical trials, to balance important covariates, most existing methods such as stratified 
permuted block design, minimization methods \citep{Taves1974,Pocock1975,Hu2012}  
and CA-BCD \citep{Antognini2011} are only for discrete covariates.  Discretizing continuous covariates 
is often less efficient and changes the nature of the covariates. A variety of 
methods for balancing continuous covariates have been proposed in the 
literature:  the methods based on ranks 
\citep{Ciolino2011,Hoehler1987,Stigsby2010}; based on p-value  
\citep{Frane1998}; based on Kullback-Leibler divergence (KLD); based on 
empirical cumulative distribution \citep{Lin2012}; based on kernel 
density \citep{Ma2013}, etc.
However, the performance of those procedures was usually evaluated by simulation studies,  their theoretical properties are not well investigated in literature. Also these methods are usually applicable for only a few covariates.

In this article, we propose a new approach --- pairwise sequential 
randomization (PSR) --- 
to generate a more balanced treatment allocation and thus to 
improve the subsequent analysis for  both causal inference 
and clinical trails settings.  
Unlike rerandomization or complete randomization, in which all units are 
allocated independently, we allocate units adaptively and sequentially by 
assigning one randomly chosen pair of units at a time.  
For each pair of units, using their 
covariate information and the existing level of imbalance of the previously 
allocated units, we adjust the probability with which the pair is allocated to 
treatment groups to avoid incidental covariate imbalance.  In this way, we are 
able to produce a much more balanced allocation of units.  The
properties of the PSR procedure are illustrated both theoretically and numerically.

The advantages of the proposed method are:
(i) For cases with a large number of covariates or a 
large number of units, the proposed method exhibits superior performance, with 
more balanced randomization and less computational time; (ii)  The PSR procedure attains 
the optimal covariate balance, in the sense that the estimated treatment effect under the proposed 
method attains its minimum variance asymptotically; and (iii) The proposed procedure is designed for directly randomizing units 
with both continuous and discrete covariates. Therefore the PSR procedure is widely applicable for balancing many important covariates in comparative studies.

This article is organized as follows.  We introduce the proposed method and 
investigate its theoretical properties in Section 
\ref{sec:Covariate_Adaptive_Maha}.  
We demonstrate 
its advantages in the treatment effect estimation and present theoretical 
properties in Section \ref{sec:Inf_Covariate_Adaptive}.  Numerical 
studies to verify the finite sample properties of the proposed method are shown in Section 
\ref{sec:Numerical_Study}.  
We further present an real data example to demonstrate the superior performance 
of our method in Section \ref{sec:Real_Data}.  
Finally, we conclude with a discussion in Section \ref{sec:Discussion} and 
relegate the outlining of proofs to Section \ref{sec:Appendix}.

\section{Pairwise Sequential Randomization}\label{sec:Covariate_Adaptive_Maha}

\subsection{Proposed Method and Its Properties}

Suppose that $n$ units (patients) are to be assigned to two treatment groups.  
Let $T_i$ 
be the assignment of the $i$-th unit, i.e., $T_i=1$ for treatment 1 and $T_i=0$ 
for treatment 2.  Consider $p$ continuous covariates for each unit.  
Let $\boldsymbol{x}_i=(x_{i1}, ...,x_{ip})^T \in \mathbb{R}^p$ represent the covariates of the 
$i$-th unit.  
Suppose all units are available for assignment at the beginning of the 
randomization.
We choose the Mahalanobis distance as the covariate imbalance measure, $M(n) = 
(\boldsymbol{\bar{x}}_1-\boldsymbol{\bar{x}}_2)^T 
\textup{cov}(\boldsymbol{\bar{x}}_1-\boldsymbol{\bar{x}}_2)^{-1}
(\boldsymbol{\bar{x}}_1-\boldsymbol{\bar{x}}_2) \propto 
(\boldsymbol{\bar{x}}_1-\boldsymbol{\bar{x}}_2)^T
\textrm{cov}(\boldsymbol{x})^{-1}
(\boldsymbol{\bar{x}}_1-\boldsymbol{\bar{x}}_2)$.
This Mahalanobis distance functions as a measure of the covariate balance 
throughout this article.  
A smaller value of $M(n)$ indicates a better covariate balance.
To assign units to treatment groups, 
we propose the following procedure, pairwise sequential randomization (PSR).
\begin{enumerate}
\item[(1)]
Arrange all $n$ units randomly into a sequence $\boldsymbol{x}_1,...,\boldsymbol{x}_n$.

\item[(2)]
Assign the first two units with $T_1=1$ and $T_2=0$.

\item[(3)]
Suppose that $2 i$ units have been assigned to treatment groups, for the 
$(2i+1)$-th and $(2i+2)$-th units:
\begin{enumerate}

\item[(3a)]
If the $(2i+1)$-th unit is assigned to treatment 1 and the $(2i+2)$-th unit to 
treatment 2, then we can calculate the ``potential'' Mahalanobis 
distance, $M_1(2i+2)$, between the updated treatment groups with $2i+2$ units.

\item [(3b)]
Similarly, if the $(2i+1)$-th unit is assigned to treatment 2 and the 
$(2i+2)$-th unit to treatment 1, then we can calculate the other 
``potential'' Mahalanobis distance, $M_2(2i+2)$.

\end{enumerate}

\item[(4)]
Assign the $(2i+1)$-th unit to treatment groups according to the following 
probabilities:
\begin{align*}
P(T_{2i+1}=1|\boldsymbol{x}_{2i},...,\boldsymbol{x}_{1},T_{2i},...,T_{1})&=
\begin{cases}
q & \text{if } M_1(2i+2) < M_2(2i+2),\\
1-q & \text{if } M_1(2i+2) > M_2(2i+2),\\
0.5 & \text{if } M_1(2i+2) = M_2(2i+2),
\end{cases}
\end{align*}
where $0.5<q<1$, and assign $T_{2i+2}=1 - T_{2i+1}$ to maintain the equal 
proportions.

\item[(5)]

Repeat the last two steps until all units are assigned.  If $n$ is odd, assign 
the last unit to two treatments with equal probabilities.
\end{enumerate}

There are several advantages for adopting Mahalanobis distance as the imbalance 
measure. 
First, it is an affinely invariant imbalance measure,  which is 
appealing especially for multivariate data. 
It is an overall imbalance measure which standardizes and aggregates each 
covariate imbalance information. 
A low Mahalanobis distance guarantees low imbalance levels in all covariates.
Note that when the covariance matrix is identity matrix, the Mahalanobis distance essentially becomes the L2 norm of the imbalance vector, which is a traditional measure of covariate imbalance in clinical trials.
In addition, using Mahalanobis distance as imbalance measure, various desirable statistical 
properties can be obtained, such as the reduction in variance of the 
estimated treatment effect and optimal asymptotic variance for treatment effect 
estimation.
In practice, the covariance matrix is replaced with the sample covariance 
matrix.

The value of $q$ is set to 0.75 throughout this article.  Different values of 
$q$ will not affect the theoretical results presented in this article. For a 
further discussion of $q$, please see Hu and Hu (2012).  
Note that the sequence in which the units are allocated is not unique.  
Rather, there are $n!$ different possible sequences, but their performances are 
similar, especially when $n$ is large.

We now study the asymptotic properties of the Mahalanobis distance, $M(n)$, 
obtained using the proposed method.

\begin{theorem}\label{thm:MahaDistLimitDist}
Under the pairwise sequential randomization (PSR), suppose that the covariate $\boldsymbol{x}_i$, $i=1,...,n$, is independent and identically distributed as a multivariate normal distribution with zero mean; then we have $M(n)=O_p(n^{-1})$.
\end{theorem}

Note that the Mahalanobis distance that is obtained through the complete 
randomization, $M_{\textup{CR}} (n)$, has a stationary distribution of a 
Chi-square distribution with $p$ degrees of freedom (regardless of $n$), i.e., 
$M_{\textup{CR}} (n) \sim \chi^2_{df=p}$.  Therefore, the Mahalanobis distance 
obtained through rerandomization,  $M_{{\textup{RR}}}(n)$, has a conditional 
Chi-square distribution, i.e., $M_{\textup{RR}} (n) \sim 
\chi^2_{df=p}|\chi^2_{df=p}<a$.  Hence, as the sample size $n$ increases, the 
proposed method reveals a greater advantage over both rerandomization and 
complete randomization, because $M(n)$ converges to 0 at the rate of $1/n$.  
That is, the more units included, the better the covariate balance becomes.

Moreover, as the number of covariates $p$ 
increases, the distribution of $M_{\textup{CR}}(n)$ becomes flatter, 
which implies poorer allocation in terms of covariate balance.  As a 
consequence, rerandomization has a lower 
probability of acceptance, $p_a=P(M_{\textup{CR}}(n)<a)$.  Therefore, the 
advantage of the proposed method also becomes more significant as $p$ 
increases, because the $M(n)$ obtained using the proposed method converges to 0 
regardless of the magnitude of $p$.

In Figure \ref{fig:rate_M_vs_n_w_1overn}, we conduct a simple simulation by 
plotting the sequences of Mahalanobis distance as more units are assigned using 
the proposed method.
As we can see, 
the trajectories converges to zero approximately at the speed of $1/n$.  
Extensive simulation studies can be found in Section \ref{sec:Numerical_Study}.

\begin{figure}  \centering
    \includegraphics[scale=0.45]{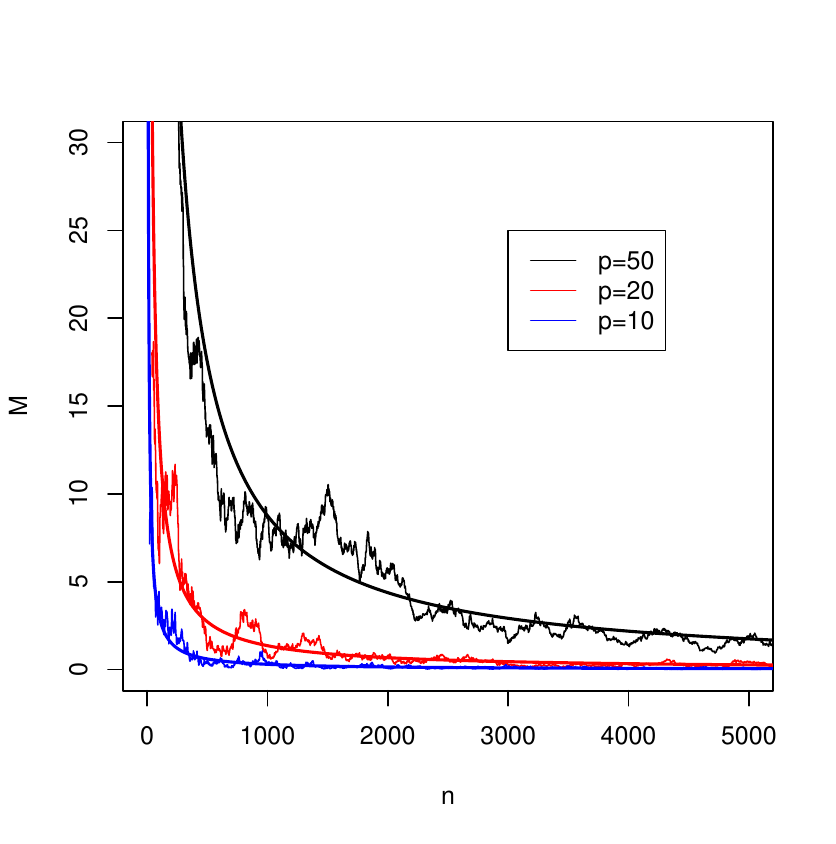}
    \caption{Convergence of $M(n)$ using the proposed 
    method.  Solid curves are fitted trends $M= c/n$ with $c>0$.}
    \label{fig:rate_M_vs_n_w_1overn}
\end{figure}

\subsection{Clinical Trial Settings}\label{sec:proposed_clinical_trial}

The proposed algorithm can be easily adopted in clinical trial studies where 
patients are sequentially enrolled and the treatment is conducted after the 
individual enrollment. Since the units come in a natural order, we do not need Step 1 anymore.
In order to have a valid covariance matrix estimate, we also need 
to increase the burn in number in Step 2 to be larger than $p$. For example, we can 
implement the first $m$ pairs ($2m>p$) by simple randomization (in each pair, the assignment is $(1,0)$ or $(0,1)$ with half probability). For the $(i+1)$th pair (here $i>m$), $2i$ units  have already been enrolled 
in study. In Step 3, one calculate  calculate the potential Mahalanobis 
distance, we only need to use the sample covariance matrix of $2i$ 
units, $\widehat{\Sigma}_{2i}$. Other steps are the same.

In literature, continuous covariates are usually discretized in order to be included in the above balancing procedures. However, breakdown of a continuous covariate into subcategories means increased effort and loss of information as pointed in \citet{Scott2002}. 
\citet{Ciolino2011} further pointed out that: ``Lack of publicity for practical methods for continuous covariate balancing and lack of knowledge on the cost of failing to balance continuous covariates results in a common phenomenon, whereby continuous covariates are excluded from the randomization plan in clinical trials.''
It is important to note that the proposed method is designed for directly randomizing units 
with continuous covariates. Also the PSR procedure works well for large $p$ and $n$, while the other methods only work for small $p$.

Through simulation studies (Section \ref{sec:Numerical_Study}), we can show 
that the above scenarios yield similar 
results in terms of covariate balance especially 
when sample size is large.

\section{Treatment Effect Estimation}\label{sec:Inf_Covariate_Adaptive}

\subsection{Framework}
After allocating the units to treatment groups, we are interested in estimating 
the treatment effect from the outcome variable $y_i$ obtained under the 
treatment $T_i$ for $i=1,...,n$. 
A natural choice is
\begin{align}\label{eq:tau_hat}
\hat{\tau}=\frac{\sum_{i=1}^{n} T_i y_i}{\sum_{i=1}^{n}T_i}-\frac{\sum_{i=1}^{n} (1-T_i)y_i}{\sum_{i=1}^{n}(1-T_i)},
\end{align}
which is simply the difference in the sample means of $y_i$ for the different 
groups.
One problem with $\hat{\tau}$ is that if there is an imbalance in the 
covariates, it will affect the accuracy of $\hat{\tau}$.  For example, if we 
estimate the effect of a drug when the treatment 1 group contains mostly males 
and the treatment 2 group contains mostly females, then the estimated treatment 
effect $\hat{\tau}$ will not be able to exclude the effect of gender.

To adjust for such an imbalance, we can use linear regression to estimate the treatment effect.  That is, conditional on the treatment assignment $T_i$, each outcome variable is assumed to follow the model below:
\begin{align}\label{eq:true_model}
y_i=\mu_1 T_i + \mu_2 (1-T_i)+\beta_1 x_{i1}+...+\beta_p x_{ip}  + \epsilon_i,
\end{align}
where $\mu_1$ and $\mu_2$ are the main effects of treatments 1 and 2, respectively, and $\mu_1-\mu_2=\tau$ is the treatment effect.  Furthermore, $\beta_j$ represents the covariate effect, and $\epsilon_i$ is an independent and identically distributed random error with zero mean and constant variance $\sigma_{\epsilon}^2$, and is independent of $\boldsymbol{x}_i = (x_{i1},..., x_{ip})^T$.  All covariates $\boldsymbol{x}_i$, $i=1,...,n$, are independent and identically distributed.

Let us define \begin{align*}
\boldsymbol{Y}=\begin{bmatrix}
   y_{1}\\
   y_{2}\\
   \vdots\\
   y_{n}
  \end{bmatrix},
\boldsymbol{X}= \begin{bmatrix}
  \boldsymbol{x}_1^T \\
  \boldsymbol{x}_2^T \\
  \vdots  \\
  \boldsymbol{x}_n^T
 \end{bmatrix}
 = \begin{bmatrix}
  x_{11} &\cdots & x_{1p} \\
  x_{21} &\cdots & x_{2p} \\
  \vdots  & \ddots & \vdots  \\
  x_{n1} &\cdots & x_{np}
 \end{bmatrix},
\boldsymbol{T}=\begin{bmatrix}
   T_{1}\\
   T_{2}\\
   \vdots\\
   T_{n}
  \end{bmatrix},
\widetilde{\boldsymbol{T}}= \begin{bmatrix}
  T_{1} & 1-T_{1} \\
  T_{2} & 1-T_{2}\\
  \vdots  & \vdots\\
  T_{n} & 1-T_{n}
 \end{bmatrix},
\end{align*}
$\boldsymbol{\widetilde{X}}=[\widetilde{\boldsymbol{T}} ; \boldsymbol{X}]$, $\boldsymbol{\beta}=(\beta_1,...,\beta_p)^T$, and $\boldsymbol{\beta}^*=(\mu_1,\mu_2,\beta_1,...,\beta_p)^T$.  Then, we can obtain the ordinary least squares estimate of $\boldsymbol{\beta}^*$:
\begin{align*}
\hat{\boldsymbol{\beta}}^*=(\widetilde{\boldsymbol{X}}^T \widetilde{\boldsymbol{X}})^{-1}\widetilde{\boldsymbol{X}}^T \boldsymbol{Y}.
\end{align*}
Let us consider $\boldsymbol{L}=(1,-1,0,...,0)^T$, a $(p+2)$-dimensional vector.  We define 
\begin{align*}
\tilde{\tau}=\boldsymbol{L}^T\hat{\boldsymbol{\beta}}^*,
\end{align*} 
which is another estimate of the treatment effect that is adjusted for the 
imbalance in the covariates.  Note that if $\widetilde{\boldsymbol{X}}$ does 
not include any covariates, i.e., 
$\boldsymbol{\widetilde{X}}=\widetilde{\boldsymbol{T}}$, then the regression 
model is $y_i=\mu_1 T_i + \mu_2 (1-T_i) + \epsilon_i$, and $\tilde{\tau}$ 
becomes $\hat{\tau}$ in Equation \eqref{eq:tau_hat}, which is the estimated 
treatment effect without adjusting for the imbalance in the covariates.

In the next section, we study the properties of $\hat{\tau}$ and $\tilde{\tau}$ 
under our proposed method (i.e., $\hat{\tau}_{\textup{PSR}}$ and 
$\tilde{\tau}_{\textup{PSR}}$) and under complete other randomization methods such as CR and RR.

\subsection{Theoretical Properties}

Under complete randomization and rerandomization, $\hat{\tau}_{\textup{CR}}$ 
and $\hat{\tau}_{\textup{RR}}$ are unbiased.  We can similarly show the consistency and asymptotic normality of $\hat{\tau}_{\textup{PSR}}$ and $\tilde{\tau}_{\textup{PSR}}$ for the proposed 
method.  
Before introducing our key properties, we first 
show the following properties:
\begin{theorem}\label{thm:covariance_covariate}
Under the pairwise sequential randomization (PSR), suppose that the covariate $\boldsymbol{x}_i$, $i=1,...,n$, is independent and identically distributed as a multivariate normal distribution with zero mean; then we have
\begin{align*}
\textup{cov}[\boldsymbol{\bar{x}}_1-\boldsymbol{\bar{x}}_2|\boldsymbol{X}, 
\textup{PSR}]=u_{n} 
\textup{cov}[\boldsymbol{\bar{x}}_1-\boldsymbol{\bar{x}}_2|\boldsymbol{X}, 
\textup{CR}],
\end{align*}
where $u_{n} = \mathbb{E}[M(n)/p|\boldsymbol{X}, \textup{PSR}]$ and $u_{n} = 
O(n^{-1})$.
\end{theorem}

In randomized experiments, the emphasis typically is placed on the percent 
reduction in variance (PRIV) defined by \citet{Morgan2012}.  This 
quantity represents the percentage by which the randomization method reduces 
the variance of the differences in the means calculated for the different 
treatment groups.  A higher value of the PRIV indicates that the 
means are closer to each other.  Consider the PRIV for the $j$-th covariate, 
\begin{align*}
100\Big(\frac{\textup{Var}[\bar{x}_{j,1}-\bar{x}_{j,2}|\boldsymbol{X}, 
\textup{CR}]-\textup{Var}[\bar{x}_{j,1}-\bar{x}_{j,2}|\boldsymbol{X}, 
\textup{PSR}]}{\textup{Var}[\bar{x}_{j,1}-\bar{x}_{j,2}|\boldsymbol{X}, 
\textup{CR}]}\Big),
\end{align*}
where $\bar{x}_{j,1}$ and $\bar{x}_{j,2}$ are the $j$-th elements of 
$\boldsymbol{\bar{x}}_1$ and $\boldsymbol{\bar{x}}_2$.  According to Theorem 
\ref{thm:covariance_covariate}, the PRIV of each covariate is $100(1-u_{n}) \%$ 
under the proposed method.  Recall that the PRIV of rerandomization for
each covariate is $100(1-v_a) \%$ where $v_a>0$ is a function of $a$.  
In contrast, for the proposed method, $\textup{PRIV}_{\textup{PSR}} \to 100 \%$ 
as $n \to 
\infty$.
which implies that, as the sample size increases, the covariate imbalance 
reaches the minimum level.  This is particularly useful when the covariates and 
outcome are correlated, because in this case, the proposed method will in turn 
improve the precision of the estimation of the treatment effect, as detailed in 
the following theorem.
\begin{theorem}\label{thm:percent_reduction_variance}
Under the pairwise sequential randomization (PSR), suppose that the outcome variable 
$y_i$ and the covariate $\boldsymbol{x}_i$ are normally distributed and that 
the treatment effect is additive; then, the percent reduction in variance 
(PRIV) of $\hat{\tau}_{\textup{PSR}}$ is $100(1-u_{n})R^2$, where $R^2$ is the 
squared multiple correlation between $y_i$ and $\boldsymbol{x}_i$ within the 
treatment groups, and $u_{n}=O(n^{-1})$.
\end{theorem}

Recall that the PRIV of $\hat{\tau}_{\textup{RR}}$ is $100(1-v_a) R^2$ 
\citep{Morgan2012}, which is a constant and does not depend on the sample 
size.  In contrast, the PRIV of $\hat{\tau}_{\textup{PSR}}$ is 
$100(1-u_{n})R^2$ and converges to $100R^2$ as the sample size $n \to \infty$.  
In fact, the PRIV of $\hat{\tau}_{\textup{PSR}}$ is simply the PRIV of the 
covariates scaled by $R^2$.  
We further plot the PRIVs of $\hat{\tau}_{\textup{PSR}}$ and of 
$\hat{\tau}_{\textup{RR}}$ (with a fixed acceptance probability of $p_a=0.05$) 
in Figure \ref{fig:PRIV_CAM_p_vs_n}.  Note that we let $R^2=1$ in both figures 
only for illustrative purposes (as in \citet{Morgan2012}).  It is evident that 
as $n$ increases, at each value of $p$, the PRIV of $\hat{\tau}_{\textup{PSR}}$ 
increases to $100\%$.  However, the PRIV of $\hat{\tau}_{\textup{RR}}$ at a 
given $p$ does not vary with different $n$.  The advantage of the proposed 
method over rerandomization is clear, especially for large $n$ and large $p$.
\begin{figure}
    \centering
    \begin{subfigure}[b]{0.45\textwidth}
    	\includegraphics[scale=0.45]{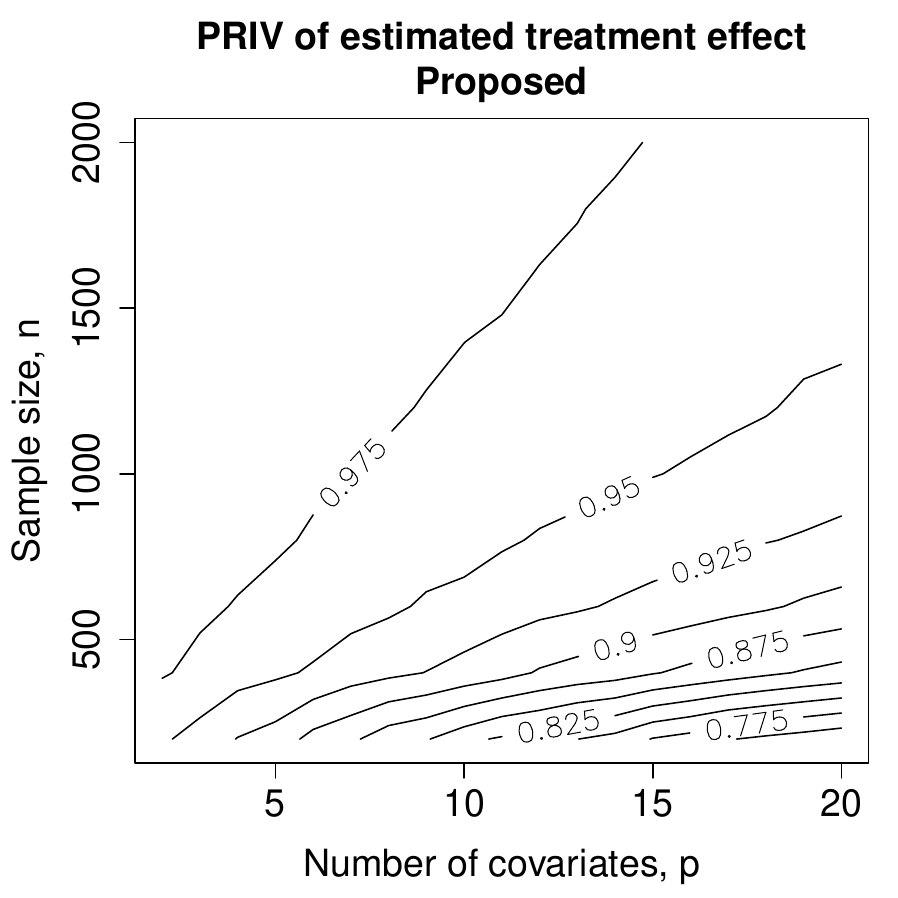}
	    \caption{Proposed method}
    \end{subfigure}
    \begin{subfigure}[b]{0.45\textwidth}
        \includegraphics[scale=0.45]{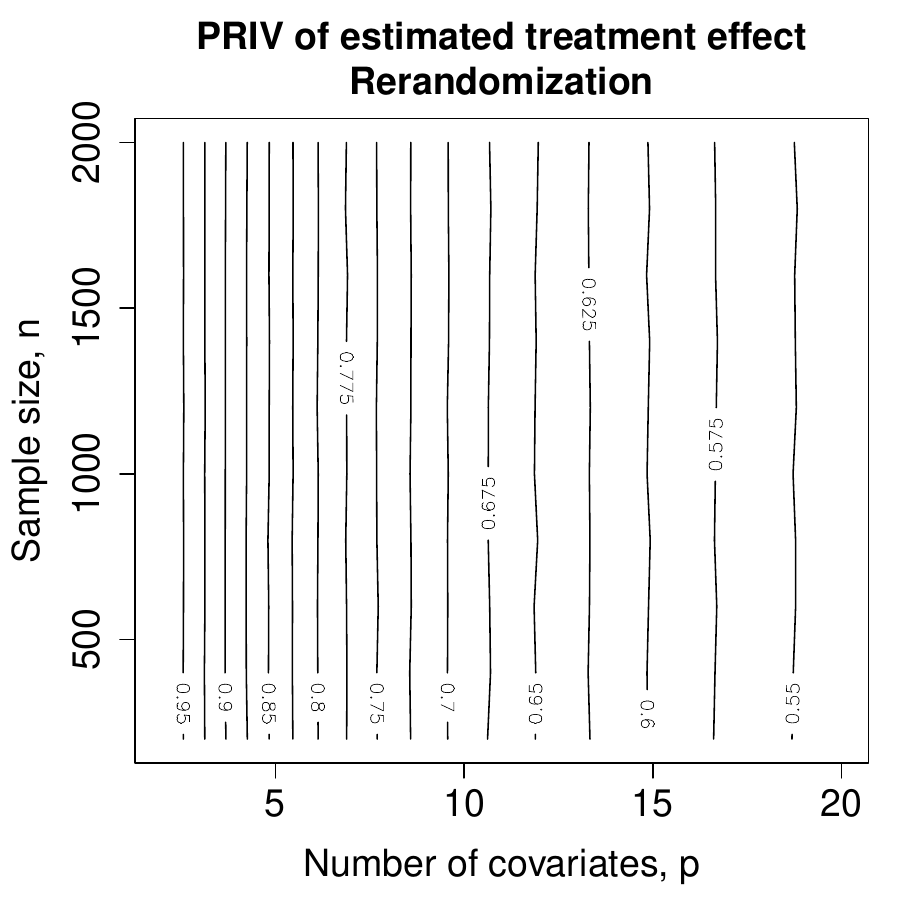}
        \caption{Rerandomization}
    \end{subfigure}
    \caption{The percent reductions in variance of the estimated treatment 
    effect under the proposed method, $\hat{\tau}_{\textup{PSR}}$, and under 
    rerandomization, $\hat{\tau}_{\textup{RR}}$, for various sample sizes and 
    numbers of covariates.  Panel (a): proposed method.  Panel (b): 
    rerandomization.}
    \label{fig:PRIV_CAM_p_vs_n}
\end{figure}

Meanwhile, the percent reduction in variance due to the adjustment via linear regression in complete randomization is $100[(1+M_{\textup{CR}}(n)/n)R^2-M_{\textup{CR}}(n)/n]$ \citep{Cox1982}, which converges to $100R^2$ as $n \to \infty$.  Therefore, we conclude that the proposed method can reduce the asymptotic variance to the minimum level.

In addition, if we further assume that the outcome variable $y_i$ truly follows 
a linear regression model, we can show that $\hat{\tau}_{\textup{PSR}}$ 
achieves the optimal precision even without adjusting for the imbalance in the 
covariates using linear regression.  That is,

\begin{theorem}[Optimal precision]\label{thm:minimum_variance} Suppose that the 
outcome variable $y_i$ follows the linear regression model in Equation 
\eqref{eq:true_model} and that we estimate the treatment effect under the 
proposed method and under complete randomization; then, we have
\begin{align*}
\sqrt{n}\big(\hat{\tau}_{\textup{PSR}}-(\mu_1-\mu_2)\big)&\overset{D}{\rightarrow}
 N(0,V_1),\\
\sqrt{n}\big(\tilde{\tau}_{\textup{PSR}}-(\mu_1-\mu_2)\big)&\overset{D}{\rightarrow}
 N(0,V_2),\\
\sqrt{n}\big(\tilde{\tau}_{\textup{CR}}-(\mu_1-\mu_2)\big)&\overset{D}{\rightarrow}
 N(0,V_3),\\
\sqrt{n}\big(\hat{\tau}_{\textup{CR}}-(\mu_1-\mu_2)\big)&\overset{D}{\rightarrow} N(0,V_4),
\end{align*}
where $4\sigma^2_{\epsilon}=V_1=V_2=V_3<V_4$.
\end{theorem}

This theorem implies that under the proposed method, the precision of the 
estimated treatment effect obtained using a simple sample mean 
difference, $\hat{\tau}_{\textup{PSR}}$, is the same as the precision of the 
estimate obtained through a linear regression which adjusts for the covariate 
imbalance, $\tilde{\tau}_{\textup{PSR}}$.  This suggests that the regression 
adjustment would not be necessary under the proposed method.
In other words, the proposed method can balance the 
covariates so well that, asymptotically, the simple sample mean difference 
$\hat{\tau}_{\textup{PSR}}$ is just as good as the linear-regression-adjusted 
estimate $\tilde{\tau}_{\textup{PSR}}$.

Furthermore, the theorem also implies that the precision of 
$\hat{\tau}_{\textup{PSR}}$ is the same as 
the precision of the estimated treatment effect obtained from a linear 
regression under complete randomization, $\tilde{\tau}_{\textup{CR}}$, which is 
considered optimal.  
Therefore, we conclude that the $\hat{\tau}_{\textup{PSR}}$ attains optimal 
precision.
Although $\tilde{\tau}_{\textup{CR}}$ and $\hat{\tau}_{\textup{PSR}}$ have the 
same precision, it is worth noting that to calculate 
$\tilde{\tau}_{\textup{CR}}$, 
it is necessary to estimate all regression coefficients $\boldsymbol{\beta}^*$, 
whereas $\hat{\tau}_{\textup{PSR}}$ is simply the sample mean difference and 
does not require the estimation of any additional coefficients.

Similarly, we present the properties of $\hat{\tau}_{\textup{RR}}$ and 
$\tilde{\tau}_{\textup{RR}}$ for comparison.  Note that all properties are 
derived under the proposed framework which is different from the framework in
\citet{Morgan2012}.  
\begin{corollary}\label{thm:RR_variance} Under the same assumptions in 
Theorem \ref{thm:minimum_variance}, suppose that we estimate the treatment 
effect under the rerandomization; then, we have
\begin{align*}
\sqrt{n}\big(\tilde{\tau}_{\textup{RR}}-(\mu_1-\mu_2)\big)&\overset{D}{\rightarrow}
 N(0,V_5),\\
\sqrt{n}\big(\hat{\tau}_{\textup{RR}}-(\mu_1-\mu_2)\big)&\overset{D}{\rightarrow}
 N(0,V_6),
\end{align*}
where $4\sigma^2_{\epsilon}=V_1=V_2=V_3=V_5<V_6<V_4$.
\end{corollary}

From the theorem above, we conclude that rerandomization cannot achieve the optimal precision in contrast to the proposed method.  
It cannot completely remove the covariate imbalance either.
In Table \ref{tab:thm_demo}, we summarize the relationships of the 
asymptotic variances of the different estimates presented by this article.  

\begin{table}
\begin{center}
\begin{tabular}{ c | c | c  c  c}
\hline
\hline
Randomized	& Randomization & \multicolumn{3}{ |c 
}{Working model for estimating $\mu_1-\mu_2$}\\
Covariates & Method	& $\texttt{lm(}\boldsymbol{Y} \sim 
\boldsymbol{\widetilde{T}}\texttt{)}$	&	& $\texttt{lm(}\boldsymbol{Y} \sim 
\widetilde{\boldsymbol{T}} + \boldsymbol{X}\texttt{)}$\\
\hline
\multirow{5}{*}{$\boldsymbol{X}$}	& CR	& Asym. Var.& $>$	& Asym. Var.\\
 &	& $\lor$	&	& $\|$\\
 & RR & Asym. Var.	&	& Asym. Var. \\
 &	& $\lor$	&	& $\|$\\
 & Proposed	& Asym. Var. & $=$& Asym. Var.\\
\hline
\end{tabular}
\caption{Demonstration of the relationship of asymptotic variances of different 
estimates.  All results are derived under the proposed framework.}
\label{tab:thm_demo}
\end{center}
\end{table}

\subsection{Computational Advantage}\label{sec:Computational_Time}
The previous section clearly demonstrates the advantages of the proposed 
method.  A natural question is whether we can also let 
$v_a \to 0$ in the rerandomization to improve its performance to match 
that of the proposed method (because rerandomization allows researchers to 
increase the power of the analysis at the expense of computational time 
\citep{Morgan2012}).  However, this option is extremely computationally 
expensive in many cases, as illustrated below.

\begin{theorem}\label{thm:acceptance_prob}
For rerandomization, to achieve the same level of covariate balance 
of the pairwise sequential randomization (PSR) (i.e., the average Mahalanobis distance under the 
PSR), the acceptance probability $p_a$ of rerandomization is 
$\chi^2_{df=p}(a^*)$, where 
$\chi^2_{df=p}(\cdot)$ is the cumulative distribution function of a Chi-square 
distribution with $p$ degrees of freedom, and $a^*$ is the root of 
$\gamma(p/2,a^*/2)Dp^2=2\gamma(p/2+1,a^*/2)n$ where $D>0$ is a constant and 
$\gamma(w,t)=\int_{0}^{t} x^{w-1} \exp\{-x\} dx$ is the incomplete gamma 
function.
\end{theorem}

%

We report the acceptance probabilities for several scenarios as quantitative 
values in Table \ref{tab:accept_prob}.  
As we can see, for a small sample size and low-dimensional covariates, the 
acceptance probability are reasonable. 
However, as either $p$ and $n$ increase, the acceptance probability approaches 
0 very fast.

\begin{table}
\begin{center}
\begin{tabular}{ c | c c c c c}
\hline
\hline
$n$ &  $p=2$  &   $p=5$  &   $p=10$   &   $p=20$    &    $p=30$ \\
\hline
1000 & 0.019360138 & 5.889118e-04 & 1.366763e-05 & 2.041414e-07 & 2.886993e-08\\
2000 & 0.009504544 & 1.058795e-04 & 4.742458e-07 & 3.091250e-10 & 2.424319e-12\\
3000 & 0.006528596 & 3.886533e-05 & 6.451756e-08 & 6.184287e-12 & 7.804135e-15\\
\hline
\end{tabular}
\caption{Acceptance probabilities of rerandomization to match the covariate 
balance produced by the proposed method for different levels of $n$ and $p$.}
\label{tab:accept_prob}
\end{center}
\end{table}

Suppose that the time to allocate one additional unit by the 
proposed method is $C(p)$ and that the time to allocate one additional unit by 
complete randomization is $R>0$.  
Note that complete randomization is not covariate-adaptive, therefore $R$ does 
not depend on $p$.
Suppose that the time to evaluate the Mahalanobis distance is $E(p)$.
Then, we have the following corollary.

\begin{corollary}\label{thm:computational_time_ratio}
To achieve the same level of covariate balance, the ratio of the average 
computational time of the proposed method to the average computational time of 
the rerandomization method is proportional to $\chi^2_{df=p}(a^*)C(p)/[E(p)R]$.
\end{corollary}

Because of the unknown properties of $C(p)$ and $E(p)$, we are unable to 
demonstrate the ratio of computational times as we did in 
Table~\ref{tab:accept_prob} for acceptance probability.  
However, we have conducted extensive simulation 
studies in the next section to demonstrate the computational advantages of the 
proposed method.

\section{Numerical Studies}\label{sec:Numerical_Study}

In this section, we perform simulation studies to verify the theoretical 
results and demonstrate the advantages of the proposed method.

\subsection{Convergence Rate}

First, we perform a simple experiment to verify the rate of convergence 
stated in Theorem \ref{thm:MahaDistLimitDist}.  
We simulate the unit's covariate $\boldsymbol{x}$ according to multivariate 
normal 
distribution $\boldsymbol{x} \sim \text{MN}(0,\boldsymbol{I})$.
Using different numbers of covariates $p$, we simulate sequences of units, 
assign them to treatment groups and record the corresponding sequences of 
Mahalanobis distances.
We repeat the procedure for 5000 times and plot the average Mahalanobis 
distance against the reciprocal of the sample size ($1/n$) in Figure 
\ref{fig:rate_M_vs_n}.  
It is clear that the expected Mahalanobis distance 
converges to 0 at the rate of $1/n$, as evidenced by the straight lines.
\begin{figure}  \centering
    \includegraphics[scale=0.4]{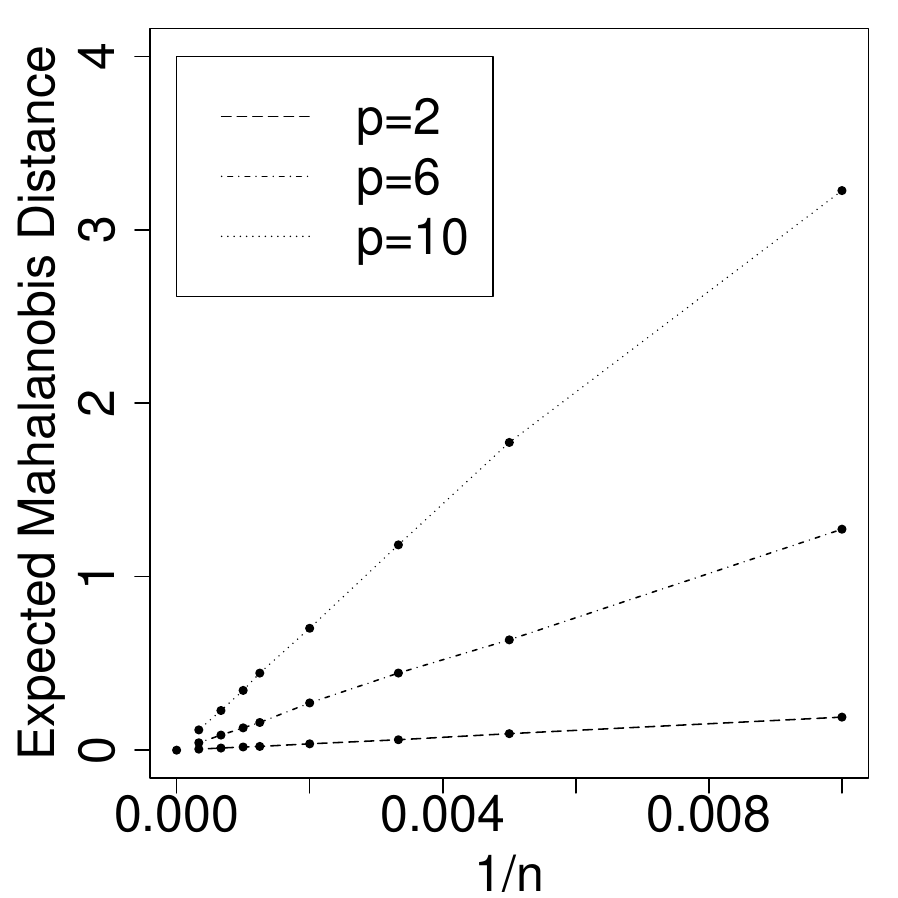}
    \caption{Verification of the rate of convergence of $M(n)$ using the 
    proposed method.}
    \label{fig:rate_M_vs_n}
\end{figure}


\subsection{Pairwise sequential randomization under Different Settings}

We also demonstrate the performance of the PSR under two different 
scenarios: (1) all units are available for assignment before the randomization 
starts, such as causal inference studies, (2) units come to the study 
sequentially and are assigned to treatment 
sequentially, such as clinical trial studies.  
In both cases, we can adopt the proposed method, the 
only difference is the number of burn in and the calculation of the sample 
covariance matrix as explained in Section \ref{sec:proposed_clinical_trial}.

We simulate the unit's covariate $\boldsymbol{x}$ according to multivariate 
normal distribution $\boldsymbol{x} \sim \text{MN}(0,\boldsymbol{I})$.
Using different $p$s and $n$s, we simulate these units, assign them to 
treatment groups and record the final Mahalanobis distances.
We plot the distributions of the Mahalanobis distance in Figure 
\ref{fig:two_scenarios}.
As the figure shows, the distributions of the Mahalanobis distance under these 
two scenarios are almost identical, especially when the sample sizes are 
large.
This is because as more and more units are assigned, the sample covariance 
matrix converges and the behaviors of Mahalanobis distance are the same for 
these two scenarios.
This simulation study verifies the applicability of the proposed method in both 
two scenarios, i.e., causal inference and clinical trial studies.

\begin{figure}[t]
    \centering
    \includegraphics[scale=0.6]{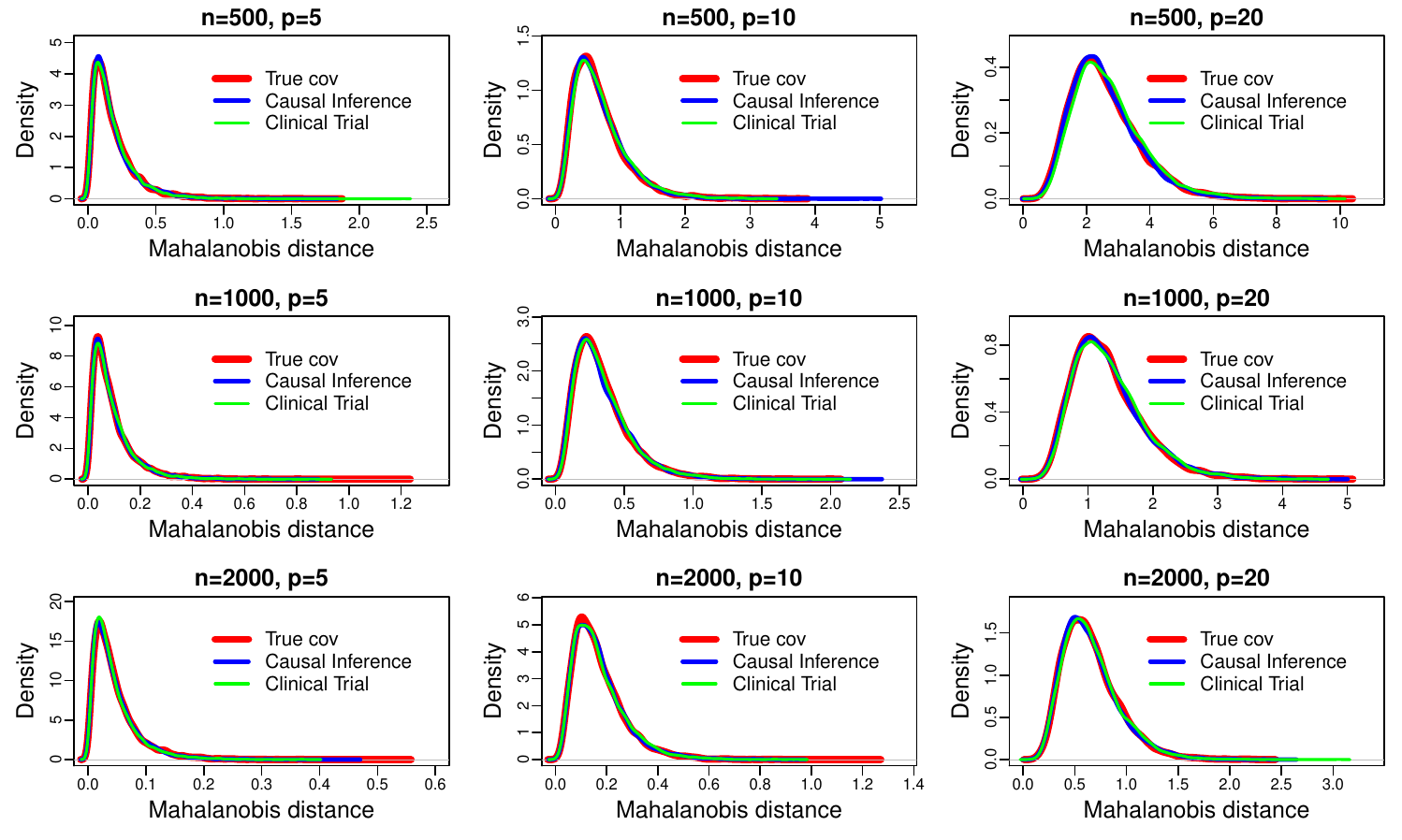}
    \caption{Comparison of the distributions of the Mahalanobis distances 
    obtained via the proposed method, $M(n)$, under three scenarios.
    Red curves are for true covariance is known.
    Blue curves are for causal inference.
    Green curves are for clinical trial.}
    \label{fig:two_scenarios}
\end{figure}

\subsection{Covariate Balance and Computational Advantage}

In this section, we compare the proposed method with other methods, especially 
rerandomization, in terms of covariate balance and computational feasibility.

We first compare the proposed method with rerandomization (with $p_a=0.05$) by 
simulating the covariates with $\boldsymbol{x} \sim 
\text{MN}(0,\boldsymbol{I})$; the results are presented in 
Figure \ref{fig:comparison_maha_dist_n_p}.  For different $n$s and $p$s, we 
plot the histograms of $M(n)$ of the proposed method and $M_{{\textup{RR}}}(n)$ 
of rerandomization.  
As the figure shows, 
as $n$ increases, the distribution of $M_{\textup{RR}} (n)$ remains unchanged, 
whereas the distribution of $M(n)$ rapidly converges to 0.  
Moreover, as $p$ increases, the distributions obtained through rerandomization 
and the proposed method become wider, but the inflation of distribution is much 
less severe for the proposed method (i.e., the overlap between the two 
distributions becomes smaller as $p$ increases).
\begin{figure}[t]
    \centering
    \includegraphics[scale=0.6]{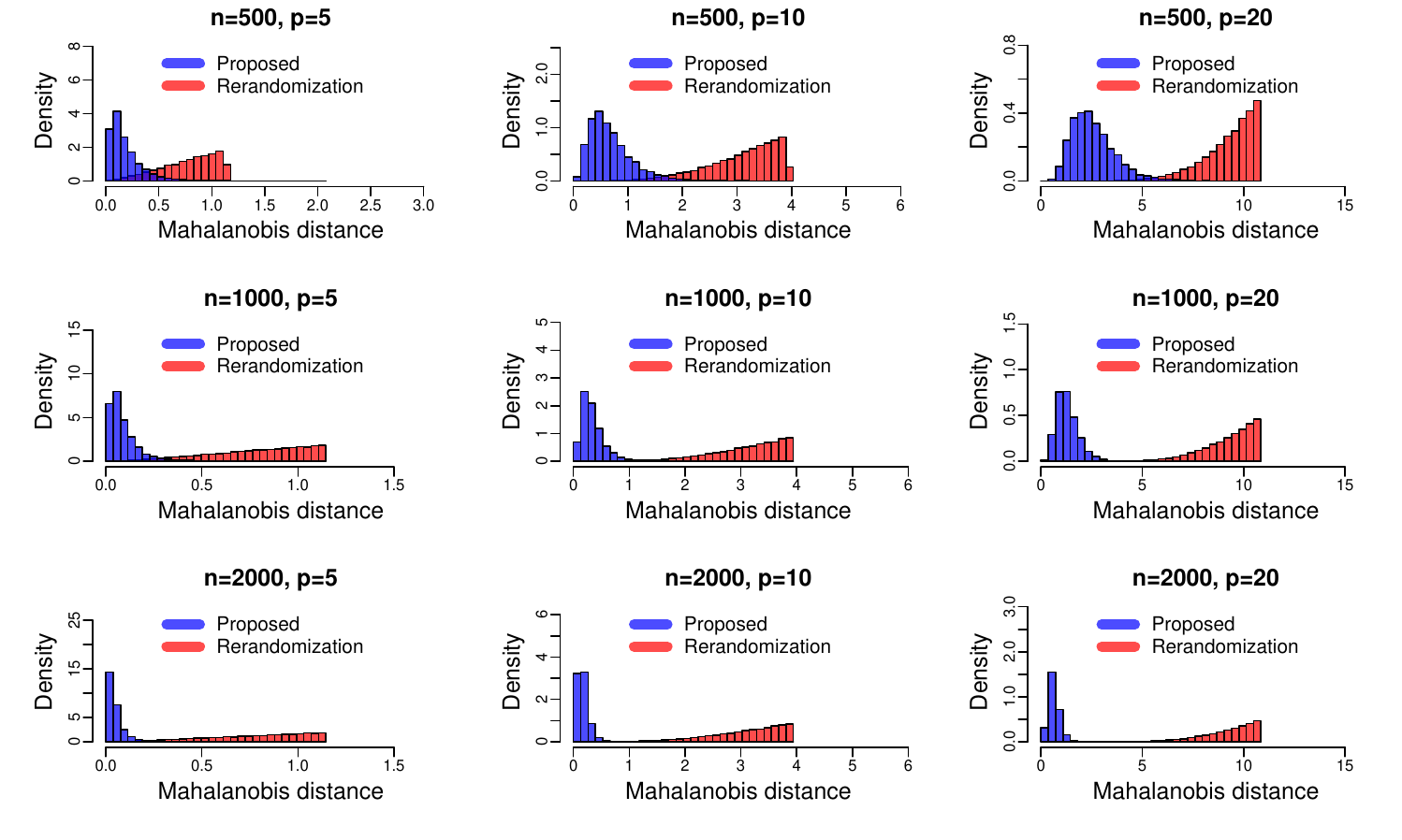}
    \caption{Comparison of the distributions of the Mahalanobis distances 
    obtained via the proposed method, $M(n)$, and rerandomization, 
    $M_{\textup{RR}}(n)$, for different sample sizes $n$ and different numbers 
    of covariates $p$.}
    \label{fig:comparison_maha_dist_n_p}
\end{figure}

Next, we compare the proposed method with rerandomization in terms 
of computational times.  Note that the proposed method only 
requires one iteration, whereas rerandomization requires multiple iterations of 
complete randomization to achieve an acceptable balance level.  
Therefore, we compared the number of 
iterations required for rerandomization to achieve the same performance (same 
Mahalanobis distance) as the proposed method.  In addition, we also compared 
the corresponding computational times.  The results are shown in Figure 
\ref{fig:comparison_computational}.  As seen in Figures 
\ref{fig:comparison_computational}a and \ref{fig:comparison_computational}b, 
when $n$ and $p$ are small, the computational 
advantage of the proposed method is not obvious.  As $n$ and $p$ increase, 
however, the proposed method gradually shows a significant advantage over 
rerandomization, because more iterations and more time are required for 
rerandomization in order to achieve the same level of performance as the 
proposed method.  As $p$ continues to increase, 
rerandomization will eventually become very computationally expensive.  In 
other words, it is nearly impossible for rerandomization to achieve the same 
performance as the proposed method.  
Note that the computational time of the proposed method grows 
only linearly with $n$ and remains the same for different $p$s, whereas the 
computational time of rerandomization grows exponentially as either $n$ or $p$ 
increases.

\begin{figure}[t]
    \centering
    \begin{subfigure}[b]{0.32\textwidth}
        \includegraphics[scale=0.28]{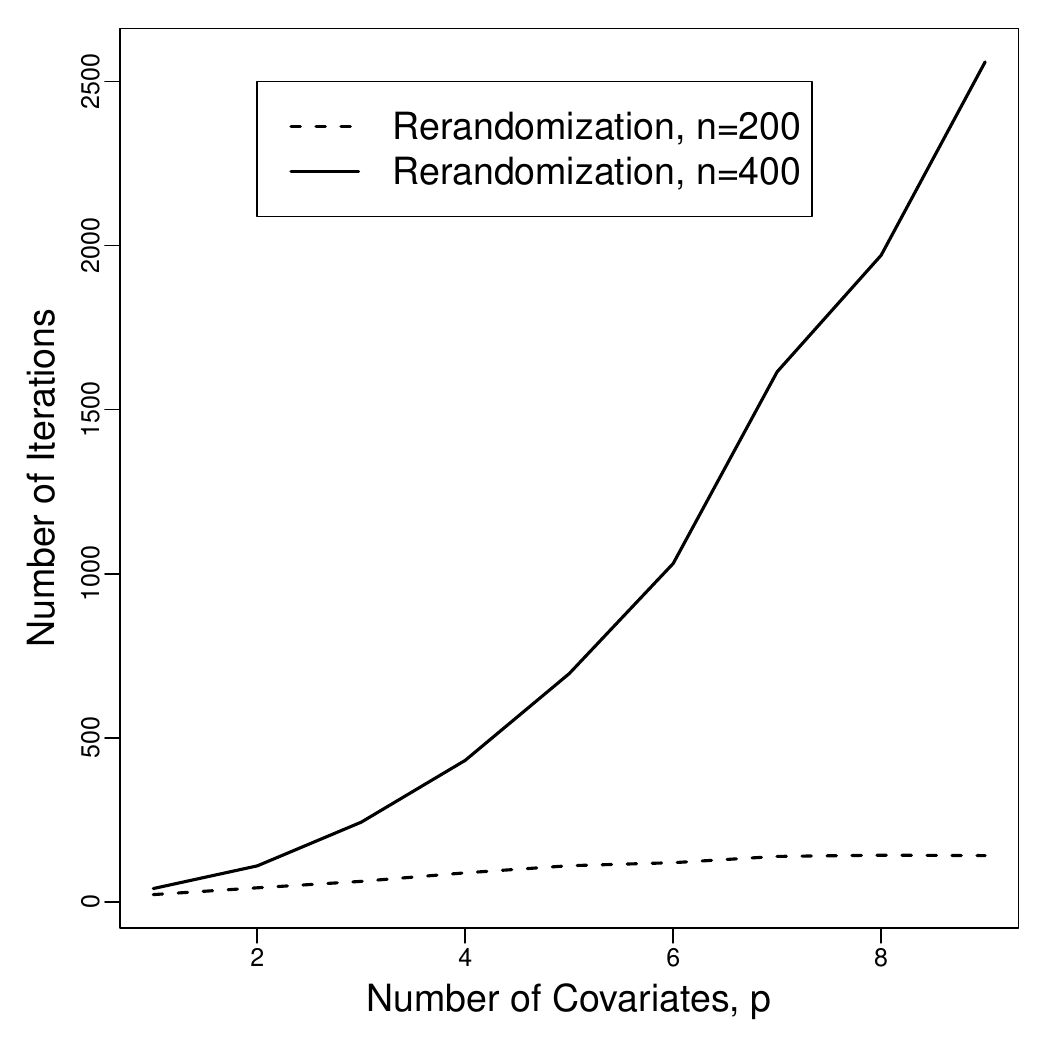}
        \caption{Number of iterations}
        \label{fig:comp_iter}
    \end{subfigure}
    \begin{subfigure}[b]{0.32\textwidth}
        \includegraphics[scale=0.28]{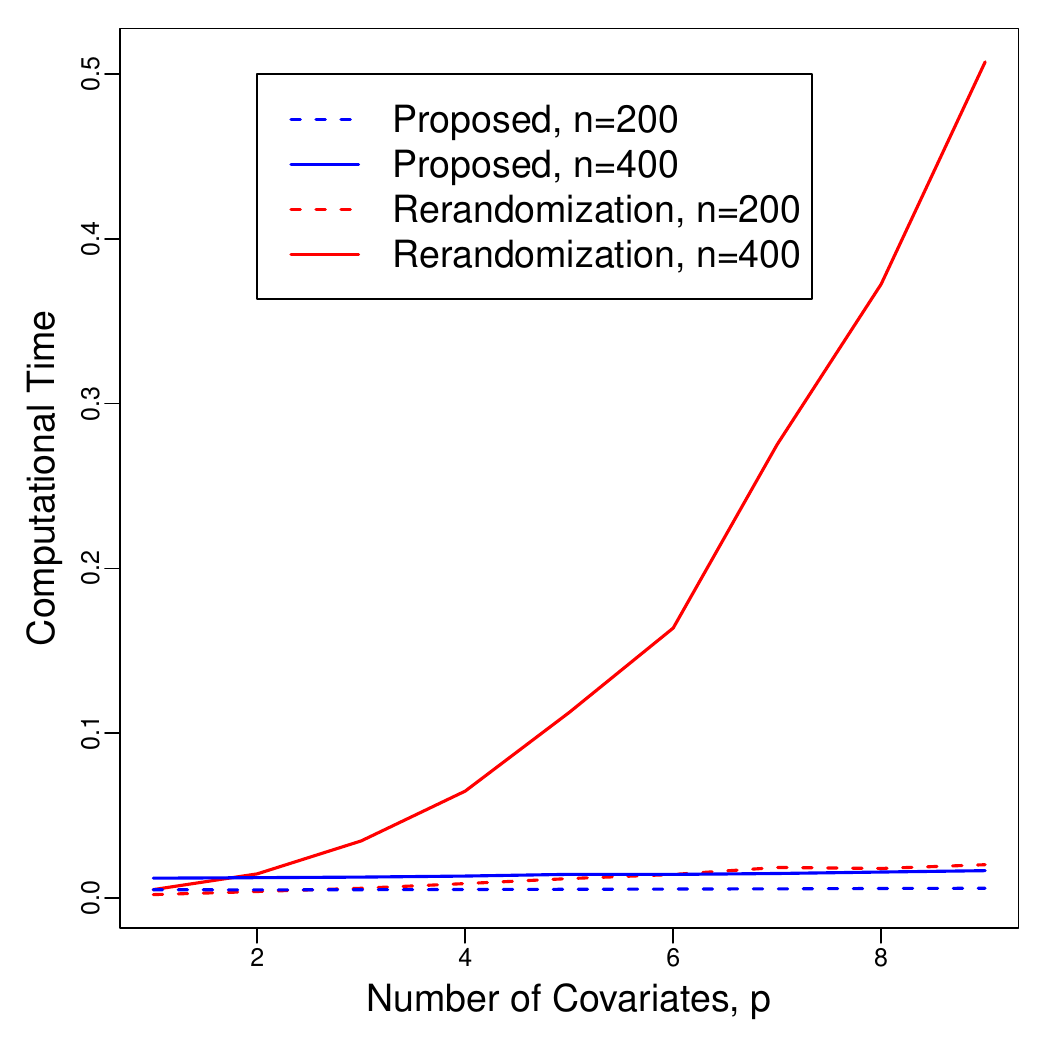}
        \caption{Computational time}
        \label{fig:comp_time}
    \end{subfigure}
    \begin{subfigure}[b]{0.32\textwidth}
        \includegraphics[scale=0.28]{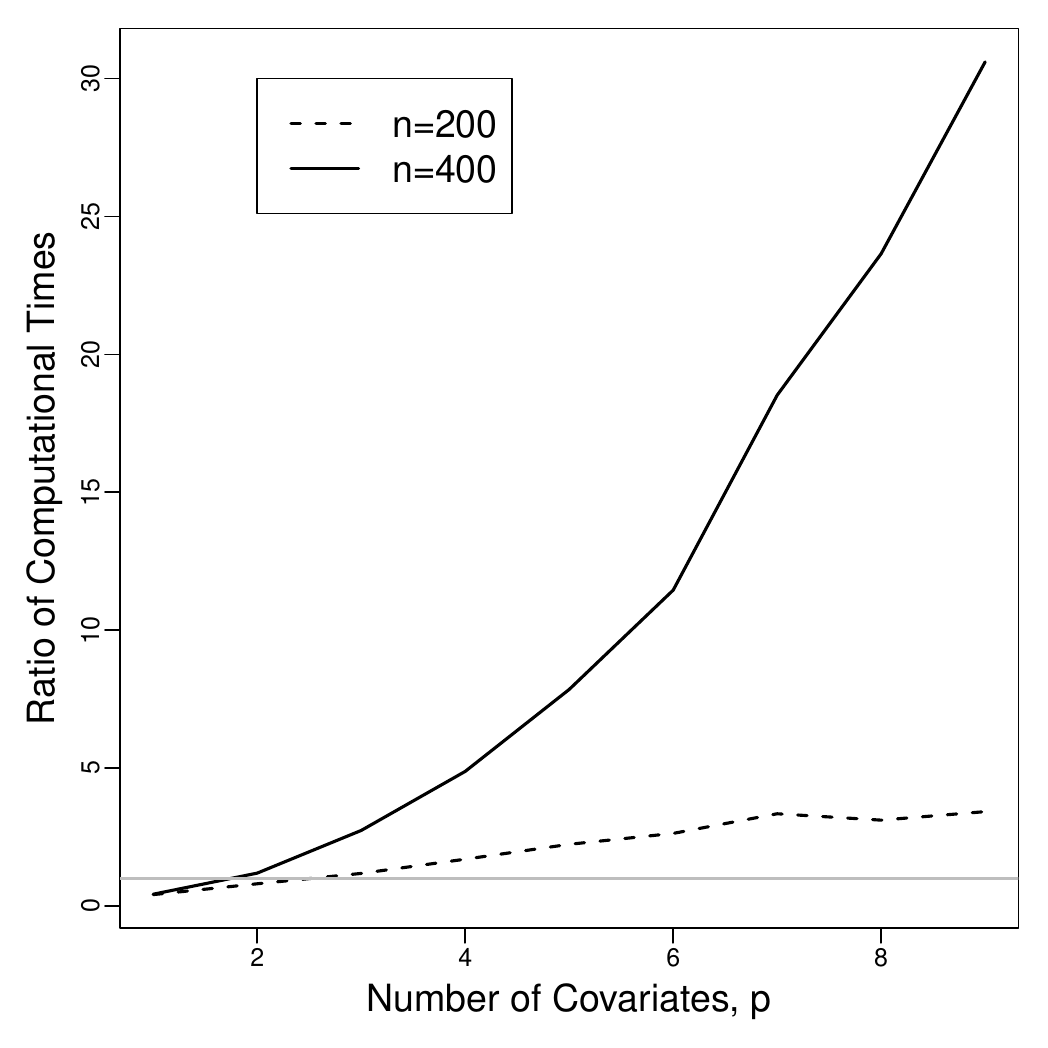}
        \caption{Ratio of times}
        \label{fig:comp_ratio}
    \end{subfigure}
    \caption{Comparison of the numbers of iterations, the computational times, 
    and the ratios of computational times for rerandomization and the proposed 
    method.  
    Panel (a):  numbers of iterations of rerandomization required to achieve 
    the same performance as the proposed method.  
    Panel (b): the corresponding computational times used in Panel (a).  
    Panel (c): the ratios of computational times shown in Panel (b).}
    \label{fig:comparison_computational}
\end{figure}

\subsection{Treatment Effect Estimation}

Finally, we compare the proposed method with other randomization 
methods in terms of estimating the treatment effect. 
We simulate ten continuous covariates $\boldsymbol{x}_i=(x_{i1},...,x_{i10})^T$ 
according to $\boldsymbol{x}_i \sim \text{MN}(\boldsymbol{0},\boldsymbol{I}_{10 
\times 10})$ with sample size of 5000.
We applied the proposed method, rerandomization and 
complete randomization to these simulated units and obtained the simulated 
treatment assignments $T_i$. 
We further simulate the outcome variable according to $y_i=\mu_1 T_i + \mu_2 
(1-T_i)+ \sum_{j=1}^{10}\beta_j x_{ij}  + \epsilon_i$, where $\mu_1=0$, 
$\mu_2=1$, 
$\beta_j=1$ for $j=1,...,10$ and $\epsilon_i \sim N(0,2^2)$.

Using the simulated data, we estimate the treatment effect using four different working models and obtain the standard error for each method under different randomization methods.
\begin{itemize}
\item[W1:]
$y_i=\mu_1 T_i + \mu_2 (1-T_i)+\epsilon_i$
\item[W2:]
$y_i=\mu_1 T_i + \mu_2 (1-T_i)+\sum_{j=1}^{3} \beta_j x_{ij} + \epsilon_i$
\item[W3:]
$y_i=\mu_1 T_i + \mu_2 (1-T_i)+\sum_{j=4}^{10} \beta_j x_{ij} + \epsilon_i$
\item[W4:]
$y_i=\mu_1 T_i + \mu_2 (1-T_i)+\sum_{j=1}^{10} \beta_j x_{ij} + \epsilon_i$
\end{itemize}
Note that W1 is equivalent to the sample mean difference $\hat{\tau}$.  
The results are presented in Table \ref{tab:thm_simu}, which is consistent with 
Table \ref{tab:thm_demo}.
As we can see, proposed method obtain the smallest standard errors among all methods.
Rerandomization performs well, but its standard errors 
are significantly larger than these of the proposed method.
Finally, not surprisingly, complete randomization has the largest standard 
errors.
As we include more covariate into the working model, the standard error 
gradually decreases.
This is because the covariate imbalance is partially adjusted by the linear regression.
When all covariates are included in the working model (i.e., W4), the standard 
errors becomes the smallest.
Note that the W4 under all randomization methods are almost the same.
This is because covariate imbalance from the randomization methods have been completely adjusted,
therefore, the standard error reaches its minimum.  
\begin{table}[t]
\begin{center}
\begin{tabular}{ c | c  c  c  c}
\hline
\hline
\multirow{2}{*}{Randomization method}	& \multicolumn{4}{ |c }{Working model for estimating $\mu_1-\mu_2$}\\
& W1 & W2 & W3 & W4\\
\hline
CR	        & 6.604616 & 5.622748 & 4.006424 & 1.970360\\
RR          & 4.036759 & 3.544364 & 2.769106 & 1.987251\\
Proposed	& 2.051219 & 2.031727 & 2.003411 & 1.985727\\
\hline
\end{tabular}
\caption{Comparison of standard errors of estimated treatment effect for working model W1, W2, W3, and W4 under different randomization methods.  This table is a verification of Table \ref{tab:thm_demo}.}
\label{tab:thm_simu}
\end{center}
\end{table}

\section{Real Data Example}\label{sec:Real_Data}

In this section, we illustrate our proposed method using a real clinical study 
of a Ceragem massage (CGM) thermal therapy bed, 
a device for treating lumbar disc disease.  
In total, there are 186 patients with $p=50$ covariates.
There are 30 continuous covariates, such as age and 
baseline measurements of the patient's current conditions, e.g., lower back 
pain and leg numbness, all measured on 0-10 scales.  
The outcome variable $y_i$, representing the measurements of the lower back 
pain after the treatment or control experiment, 
was recorded to study the treatment effect.

In the original study, the patients were randomly assigned to 
the treatment or control groups. The corresponding Mahalanobis distance was 
57.67, which indicates a moderate covariate imbalance.
To compare, we repeatedly assigned these patients to treatment groups using the 
proposed method, complete randomization, and 
rerandomization ($M<a$ and $a = 20, 30, 40$).  
The corresponding Mahalanobis distances are plotted in 
Figure~\ref{fig:comparison_M_real_data}.
Note that, in the right panel of 
Figure~\ref{fig:comparison_M_real_data}, we replicated the data four 
times to $n=744$ (which attempts to mimic the large sample size settings).

\begin{figure}[t]
    \centering
    \begin{subfigure}[b]{0.48\textwidth}
        \includegraphics[scale=0.45]{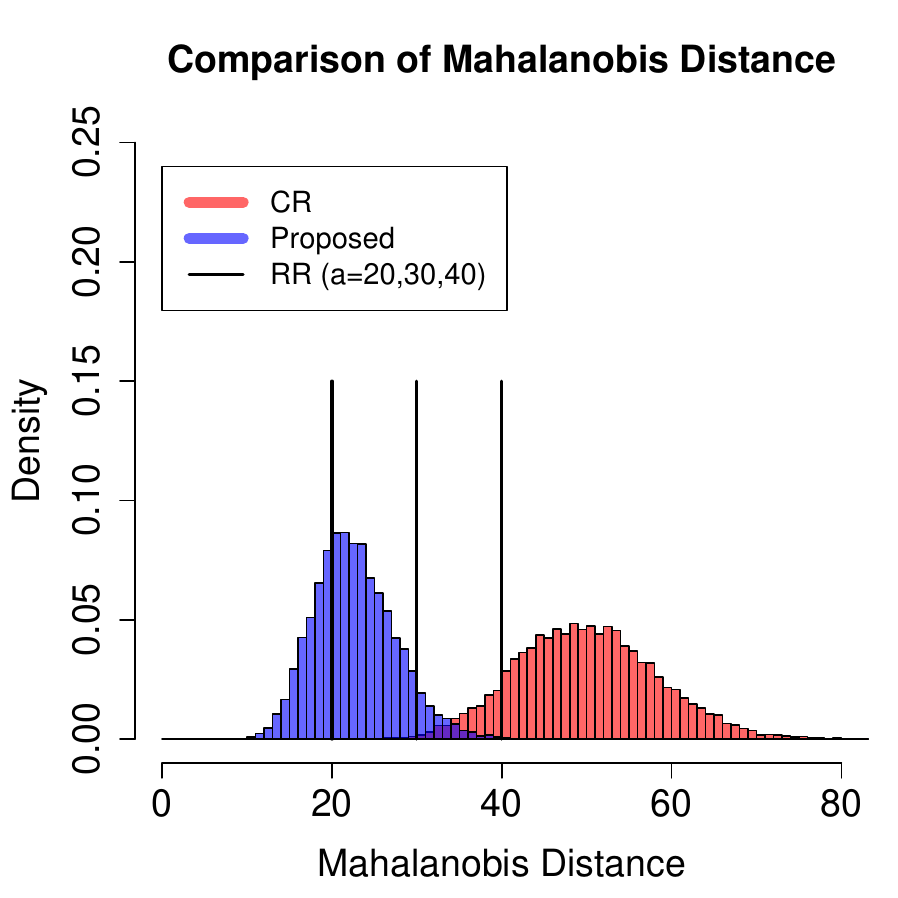}
        \caption{$n=186$}
        \label{fig:original_realdata}
    \end{subfigure}
    \begin{subfigure}[b]{0.48\textwidth}
        \includegraphics[scale=0.45]{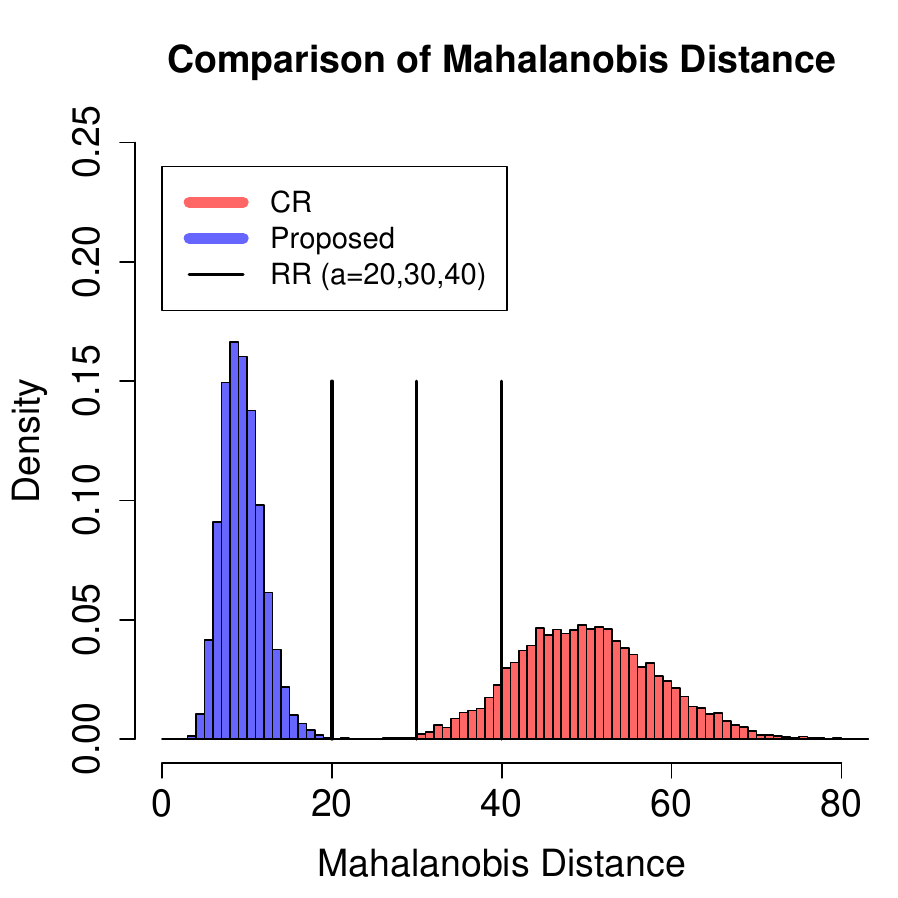}
        \caption{$n=744$}
        \label{fig:repeated_realdata}
    \end{subfigure}
    \caption{Comparison of the distributions of the Mahalanobis distance 
    obtained using the proposed method, complete 
    randomization, and 
    rerandomization.  Note that rerandomization is represented by the portion 
    of the complete randomization distribution that lies to the left of the 
    vertical line ($M=20, 30, 40$).}
    \label{fig:comparison_M_real_data}
\end{figure}

As seen from the figure, the Mahalanobis distances of the 
proposed method on the original data ($n=186$) are consistently lower than 
those of complete randomization.
If we had $n=744$ patients, the Mahalanobis distance of the proposed method 
further decreases toward 0.  
Few allocations of complete randomization could achieve the 
same level of balance as the proposed.  
Rerandomization produces the Mahalanobis distances to 
the left of the vertical lines ($M=20,30,40$), which are still not 
comparable with the proposed.

For each randomization method, we further simulated the outcome variable 
$y_i^{\textup{sim}}$ according to $y_i^{\textup{sim}}=\hat{\mu}_1 
T_i^{\textup{sim}} + \hat{\mu}_2 
(1-T_i^{\textup{sim}})+\boldsymbol{x}_i^T\hat{\boldsymbol{\beta}}+\epsilon_i^{\textup{sim}}$,
 where $T_i^{\textup{sim}}$ is the simulated patient allocation, 
$\epsilon_i^{\textup{sim}}$ is sampled from the residuals of the 
regression fitted to the original data.  $\hat{\mu}_1$, $\hat{\mu}_2$, and 
$\hat{\boldsymbol{\beta}}$ are the corresponding 
estimated regression coefficients.  

Using the simulated outcome variable, we obtained the average treatment effect 
using $\hat{\tau}$.  
The performance comparison is summarized in Table \ref{tab:real_data}.  
The proposed method exhibits the best performance compared with other methods 
especially under large sample size.  
It yields the largest PRIV and the lowest variance.  For rerandomization, a 
smaller threshold results in better performance; 
however, this comes at the cost of a longer computational time and a lower 
acceptance probability.  Note that the $R^2$ for the regression fitted to the 
original data is only $0.33$, therefore, the maximum of PRIV is 0.33.  
Because of the finite sample size, the optimal PRIV cannot be achieved.  
We can see that if we increase the sample 
size, the PRIV of the proposed method is greatly improved and is close to 
optimal, whereas that of the rerandomization method does not improve at all.  
The gain from the proposed method is quite substantial.

\begin{table}[t]
\begin{center}
\begin{tabular}{ c | c | c  c  c }
\hline
\hline
Sample Size & Method & PRIV & MSE (or Var) & $u_{n}$ or $v_a$\\
\hline
\multirow{5}{*}{$n=186$}	& Proposed & 19.7\% & 0.081 & 0.502\\
							& RR ($M<40$) & 12.2\%	& 0.090 & 0.730\\
							& RR ($M<30$) & 15.1\%	& 0.085 & 0.562\\
							& RR ($M<20$) & 20.3\%	& 0.081 & 0.501\\
							& CR & - & 0.100 & -\\
\hline
\multirow{5}{*}{$n=744$}	& Proposed    & 27.4\%  & 0.018 & 0.205\\
							& RR ($M<40$) & 10.9\%	& 0.022 & 0.718\\
							& RR ($M<30$) & 14.6\%	& 0.021 & 0.556\\
							& RR ($M<20$) & 20.6\%	& 0.018 & 0.380\\
							& CR & - & 0.025 & -\\
\hline
\end{tabular}
\caption{Comparison of the proposed method with rerandomization and complete randomization for real data analysis.}
\label{tab:real_data}
\end{center}
\end{table}

\section{Discussion}\label{sec:Discussion}
In this article, we have introduced a new randomization procedure for balancing 
the covariates to improve the estimation accuracy for causal inference and 
clinical trials.  Compared with 
traditional methods, the proposed method can cope with a large number of 
covariates and a large sample size, which is especially advantageous in the era 
of big data.  The 
proposed method also shows superior performance in terms of computational 
time.  In addition, it achieves optimality under the linear regression 
framework, in the sense that, asymptotically, the proposed method can balance 
the covariates so well that the imbalance adjustment provided by linear 
regression is not needed.

Although the proposed method is different from the minimization methods 
\citep{Wei1978,Begg1980, Smith1984, Smith1984b}, it can be 
extended to such settings.  
Instead of selecting a pair of units, we can select 
only one unit to allocate.  However, the behavior of the Mahalanobis distance 
in such a scenario will be further complicated, because the proportion of the 
treatment group (i.e., $\sum_{i=1}^n T_i/n$) then becomes a random variable.  
We believe that the allocation procedure should be slightly modified such that 
both the Mahalanobis distance and the proportion are controlled.  In such a 
scenario, we anticipate that the rate of convergence of the Mahalanobis 
distance can be further improved. We leave this possibility as a topic for 
future investigation.  

The proposed method is following the similar spirit of the minimization methods 
used in clinical trials \citep{Taves1974,Pocock1975,Hu2012}.  
However, the focus and context of these methods are different from ours.
Their methods are applicable for patients sequentially enrolled in a clinical trial.  
On the other hand, our proposed method can be applied {\it both} in clinical  trials 
where units enrolled sequentially, and also in causal inference where all units 
are collected before the randomization and experiment starts.
Another significant difference is that the minimization 
methods are suitable for discrete covariates, minimizing the margin and 
stratum imbalance.  The proposed method, in contrast, is suitable for both discrete and continuous covariates.

Throughout the article, we have focused on equal proportion allocation.  
However, the proposed method can be easily extended to accommodate unequal 
proportions.  For example, to achieve a ratio of $1:2$, in each iteration, we 
can allocate three units at a time to maintain the targeted proportions.

Many other potential directions for further research remain as well. For example, we have shown the optimality of the estimated treatment effect.  An extension to hypothesis testing is also of interest \citep{Ma2015}.  The optimality of the estimator hints at the most powerful test for the treatment effect.  In addition, as the number of covariates increases, it is more efficient to balance only the most important covariates \citep{Morgan2015}; therefore, the selection of the important covariates to balance in our proposed framework is an interesting topic. The proposed method may also be applied to balance important covariates in the field of crowdsourced-internet experimentation. 

\bigskip
\ssection{APPENDIX}\label{sec:Appendix}
\bigskip
We provide outlines of the key proofs in the Appendix.  The supplementary materials contain detailed proofs of all theorems.

\begin{proof}[Proof of Theorem \ref{thm:MahaDistLimitDist}]

We first convert the covariates to canonical form \citep{Rubin1992}.  Let 
$\Sigma=\text{cov}(\boldsymbol{x})$ and 
$\boldsymbol{z}_i=\Sigma^{-1/2}\boldsymbol{x}_i$ where $\Sigma^{-1/2}$ is the 
Cholesky square root of $\Sigma^{-1}$.  Suppose that $n$ is even.  By the 
assumption, $\text{cov}(\boldsymbol{z}_i)=\boldsymbol{I}$, and
\begin{align*}
M^*(n)&=n 
p_n(1-p_n)(\boldsymbol{\bar{z}}_1-\boldsymbol{\bar{z}}_2)^T(\boldsymbol{\bar{z}}_1-\boldsymbol{\bar{z}}_2).
\end{align*}
We further define
\begin{align*}
\boldsymbol{y}_n&=\frac{n}{2}(\boldsymbol{\bar{z}}_1-\boldsymbol{\bar{z}}_2)=\sum_{i:T_i = 1}\boldsymbol{z}_{i}-\sum_{i:T_i = 0}\boldsymbol{z}_{i},\\
\boldsymbol{\Delta}_{n+2}&=(-1)^{T_{n+2}}(\boldsymbol{z}_{n+1}-\boldsymbol{z}_{n+2}).
\end{align*}
We can see that $\{\boldsymbol{y}_n, \boldsymbol{y}_{n+2} ,\boldsymbol{y}_{n+4} ...\}$ is a Markov process and $\boldsymbol{y}_{n+2}=\boldsymbol{y}_{n}+\boldsymbol{\Delta}_{n+2}$.  Define the test function $V(\boldsymbol{y}_{n})=\boldsymbol{y}^T_{n}\boldsymbol{y}_{n}$.  By denoting $\mathbb{E}[~\cdot~|\boldsymbol{y}_n]=\mathbb{E}_n[~\cdot~]$, we have
\begin{align*}
\mathbb{E}_n[V(\boldsymbol{y}_{n+2})]-V(\boldsymbol{y}_{n})=\mathbb{E}_n[\boldsymbol{y}^T_{n+2}\boldsymbol{y}_{n+2}]-\boldsymbol{y}^T_{n}\boldsymbol{y}_{n}=2\mathbb{E}_n[\boldsymbol{y}^T_{n}\boldsymbol{\Delta}_{n+2}]+\mathbb{E}_n[\boldsymbol{\Delta}_{n+2}^T\boldsymbol{\Delta}_{n+2}],
\end{align*}
where $\mathbb{E}_n[\boldsymbol{\Delta}_{n+2}^T\boldsymbol{\Delta}_{n+2}]=\mathbb{E}_n[(-1)^{2T_{n+2}}(\boldsymbol{z}_{n+1}-\boldsymbol{z}_{n+2})^T(\boldsymbol{z}_{n+1}-\boldsymbol{z}_{n+2})]$ is a positive constant.  For the first term on the right, we have
\begin{align*}
\mathbb{E}_n [\boldsymbol{y}^T_{n}\boldsymbol{\Delta}_{n+2}] &=\mathbb{E}_n [\boldsymbol{y}^T_{n}(-1)^{T_{n+2}}(\boldsymbol{z}_{n+1}-\boldsymbol{z}_{n+2})]\\
&=\mathbb{E}_n \Big\{\mathbb{E} \Big[\boldsymbol{y}^T_{n}(-1)^{T_{n+2}}(\boldsymbol{z}_{n+1}-\boldsymbol{z}_{n+2})\Big|\boldsymbol{z}_{n+1},\boldsymbol{z}_{n+2}\Big]\Big\}\\
&=\mathbb{E}_n \Big\{ (1-2q)|\boldsymbol{y}^T_{n}(\boldsymbol{z}_{n+1}-\boldsymbol{z}_{n+2})| \Big\}\\
&=\mathbb{E}_n \Big\{ (1-2q)|\boldsymbol{y}_{n}|~|\boldsymbol{z}_{n+1}-\boldsymbol{z}_{n+2}|~|\cos\theta| \Big\}\\
&=(1-2q)|\boldsymbol{y}_{n}|~\mathbb{E}_n[|\boldsymbol{z}_{n+1}-\boldsymbol{z}_{n+2}|]~\mathbb{E}_n[|\cos\theta|],
\end{align*}
where $\theta$ is the angle between $\boldsymbol{y}_n$ and 
$\boldsymbol{z}_{n+1}-\boldsymbol{z}_{n+2}$.  Note that 
$\mathbb{E}_n[|\boldsymbol{z}_{n+1}-\boldsymbol{z}_{n+2}|]$ and 
$\mathbb{E}_n[|\cos\theta|]$ are two positive constants.  Since $1-2q<0$, there 
exist a constant $b>0$ and $c<0$, such as when $|\boldsymbol{y}_{n}|>b$, 
$\mathbb{E}_n[ 
\boldsymbol{y}^T_{n}\boldsymbol{\Delta}_{n+2}]+\mathbb{E}_n[\boldsymbol{\Delta}_{n+2}^T\boldsymbol{\Delta}_{n+2}]<c$.
  Therefore, $\mathbb{E}_n[V(\boldsymbol{y}_{n+2})]-V(\boldsymbol{y}_{n})<c$ 
for $|\boldsymbol{y}_{n}|>b$.  Similarly, we have 
$\mathbb{E}_n[V(\boldsymbol{y}_{n+2})]-V(\boldsymbol{y}_{n})<\mathbb{E}_n[\boldsymbol{\Delta}_{n+2}^T\boldsymbol{\Delta}_{n+2}]$
 for $|\boldsymbol{y}_{n}|\leq b$.  By the ``drift conditions'' 
\citep{Meyn2009}, we know $\boldsymbol{y}_{n}$ has a stationary distribution.  
Therefore, $n M^*(n)/(4p_n(1-p_n))=\boldsymbol{y}^T_{n}\boldsymbol{y}_{n}$ has 
a 
stationary distribution and $M^*(n)=O_p(n^{-1})$.

In practice, covariance materix $\Sigma$ is not known and is sequentially 
estimated as $\widehat{\Sigma}_n$.  We write the Mahalanobis distance as
\begin{align*}
M(n)
=&n 
p_n(1-p_n)(\boldsymbol{\bar{x}}_1-\boldsymbol{\bar{x}}_2)^T\widehat{\Sigma}_n^{-1}
(\boldsymbol{\bar{x}}_1-\boldsymbol{\bar{x}}_2).\\
=&n p_n(1-p_n)(\boldsymbol{\bar{x}}_1-\boldsymbol{\bar{x}}_2)^T
(\widehat{\Sigma}_n^{-1}+\Sigma^{-1}-\Sigma^{-1})
(\boldsymbol{\bar{x}}_1-\boldsymbol{\bar{x}}_2).\\
=&M^*(n)+n p_n(1-p_n)(\boldsymbol{\bar{x}}_1-\boldsymbol{\bar{x}}_2)^T
(\widehat{\Sigma}_n^{-1}-\Sigma^{-1})
(\boldsymbol{\bar{x}}_1-\boldsymbol{\bar{x}}_2)\\
=&M^*(n)+n p_n(1-p_n)(\boldsymbol{\bar{x}}_1-\boldsymbol{\bar{x}}_2)^T
\Sigma^{-1/2}\Sigma^{1/2}(\widehat{\Sigma}_n^{-1}-\Sigma^{-1})\Sigma^{1/2}\Sigma^{-1/2}
(\boldsymbol{\bar{x}}_1-\boldsymbol{\bar{x}}_2)\\
=&M^*(n)+n p_n(1-p_n)(\boldsymbol{\bar{z}}_1-\boldsymbol{\bar{z}}_2)^T
(\Sigma^{1/2}\widehat{\Sigma}_n^{-1}\Sigma^{1/2}-I)
(\boldsymbol{\bar{z}}_1-\boldsymbol{\bar{z}}_2)
\end{align*}
Note that $(\boldsymbol{\bar{z}}_1-\boldsymbol{\bar{z}}_2)^T
(\Sigma^{1/2}\widehat{\Sigma}_n^{-1}\Sigma^{1/2}-I)
(\boldsymbol{\bar{z}}_1-\boldsymbol{\bar{z}}_2)$ can be considered as the 
weighted norm of $\boldsymbol{\bar{z}}_1-\boldsymbol{\bar{z}}_2$, i.e., 
$||\boldsymbol{\bar{z}}_1-\boldsymbol{\bar{z}}_2||_W$ where 
$W=\Sigma^{1/2}\widehat{\Sigma}_n^{-1}\Sigma^{1/2}-I
=\Sigma^{1/2}(\widehat{\Sigma}_n^{-1} - \Sigma )\Sigma^{1/2}$.  Since 
$\widehat{\Sigma}_n \overset{p}{\to} \Sigma$, we 
have $M(n)=O_p(M(n))$
\end{proof}

\begin{proof}[Proof of Theorem \ref{thm:minimum_variance}]
We first convert the covariates to canonical form \citep{Rubin1992}.  Let $\Sigma=\text{cov}(\boldsymbol{x})$ and $\boldsymbol{z}_i=\Sigma^{-1/2}\boldsymbol{x}_i$ where $\Sigma^{-1/2}$ is the Cholesky square root of $\Sigma^{-1}$.  Suppose $n$ is even.  By the assumption of normality, $\boldsymbol{z}_i\sim N(0,\boldsymbol{I})$.  

Define
\begin{align*}
\boldsymbol{Y}=\begin{bmatrix}
   y_{1}\\
   y_{2}\\
   \vdots\\
   y_{n}
  \end{bmatrix},
\boldsymbol{Z}= \begin{bmatrix}
  \boldsymbol{z}_1^T \\
  \boldsymbol{z}_2^T \\
  \vdots  \\
  \boldsymbol{z}_n^T
 \end{bmatrix},
\boldsymbol{T}=\begin{bmatrix}
   T_{1}\\
   T_{2}\\
   \vdots\\
   T_{n}
  \end{bmatrix},
\widetilde{\boldsymbol{T}}= \begin{bmatrix}
  T_{1} & 1-T_{1} \\
  T_{2} & 1-T_{2}\\
  \vdots  & \vdots\\
  T_{n} & 1-T_{n}
 \end{bmatrix},
\end{align*}
and $\widetilde{\boldsymbol{Z}}=[\widetilde{\boldsymbol{T}}; \boldsymbol{Z}]$,  $\boldsymbol{\gamma}=(\gamma_1,...,\gamma_p)^T=(\Sigma^{-1/2})^T \boldsymbol{\beta}$, $\boldsymbol{\mu}=(\mu_1,\mu_2)^T$ and $\boldsymbol{\gamma}^*=(\boldsymbol{\mu}^T,\boldsymbol{\gamma}^T)^T=(\mu_1,\mu_2,\gamma_1,...,\gamma_p)^T$.    

Then true model, equation \eqref{eq:true_model}, can be rewritten as
\begin{align*}
\boldsymbol{Y} &=\widetilde{\boldsymbol{X}} \boldsymbol{\beta}^* + \pmb\epsilon = \widetilde{\boldsymbol{T}}\boldsymbol{\mu}+ \boldsymbol{X}\boldsymbol{\beta}+\boldsymbol{\epsilon} = \widetilde{\boldsymbol{T}}\boldsymbol{\mu}+ \boldsymbol{Z}\boldsymbol{\gamma}+\boldsymbol{\epsilon} = \widetilde{\boldsymbol{Z}} \boldsymbol{\gamma}^* + \pmb\epsilon.
\end{align*}

{\bf Part I}: $\hat{\tau}_{\textup{PSR}}$

Suppose $\boldsymbol{K}=(1,-1)$, then $\hat{\tau}_{\textup{PSR}}$ can be 
obtained by running the regression, $\boldsymbol{Y} = 
\widetilde{\boldsymbol{T}}\boldsymbol{\mu}+\boldsymbol{\epsilon}$, even though 
the true model is $\boldsymbol{Y} = 
\widetilde{\boldsymbol{Z}}\boldsymbol{\gamma}^*+\boldsymbol{\epsilon}= 
\widetilde{\boldsymbol{T}}\boldsymbol{\mu}+ 
\boldsymbol{Z}\boldsymbol{\gamma}+\boldsymbol{\epsilon}$.  In particular, we 
can write $\hat{\tau}_{\textup{PSR}}$ as
\begin{align*}
\hat{\tau}_{\textup{PSR}}&=\frac{\sum_{i=1}^{n} T_i 
y_i}{\sum_{i=1}^{n}T_i}-\frac{\sum_{i=1}^{n} 
(1-T_i)y_i}{\sum_{i=1}^{n}(1-T_i)}\\
&= \boldsymbol{K}\bigg(\frac{\widetilde{\boldsymbol{T}}^T\widetilde{\boldsymbol{T}}}{n}\bigg)^{-1}\frac{\widetilde{\boldsymbol{T}}^T\boldsymbol{Y}}{n}\\
&= \boldsymbol{K}\bigg(\frac{\widetilde{\boldsymbol{T}}^T\widetilde{\boldsymbol{T}}}{n}\bigg)^{-1}\frac{\widetilde{\boldsymbol{T}}^T(\widetilde{\boldsymbol{Z}} \boldsymbol{\gamma}^* + \boldsymbol{\epsilon})}{n}\\
&= \boldsymbol{K}\bigg(\frac{\widetilde{\boldsymbol{T}}^T\widetilde{\boldsymbol{T}}}{n}\bigg)^{-1}\frac{\widetilde{\boldsymbol{T}}^T(\widetilde{\boldsymbol{T}} \boldsymbol{\mu} + \boldsymbol{Z} \boldsymbol{\gamma}+ \boldsymbol{\epsilon})}{n}\\
&=\boldsymbol{K} \bigg[ \boldsymbol{\mu} + \bigg(\frac{\widetilde{\boldsymbol{T}}^T\widetilde{\boldsymbol{T}}}{n}\bigg)^{-1}\frac{\widetilde{\boldsymbol{T}}^T(\boldsymbol{Z} \boldsymbol{\gamma}+ \boldsymbol{\epsilon})}{n}\bigg]\\
&=\mu_1-\mu_2 + \boldsymbol{K}\bigg(\frac{\widetilde{\boldsymbol{T}}^T\widetilde{\boldsymbol{T}}}{n}\bigg)^{-1}\frac{\widetilde{\boldsymbol{T}}^T(\boldsymbol{Z} \boldsymbol{\gamma}+ \boldsymbol{\epsilon})}{n}.
\end{align*}
We know, as $n \to \infty$, 
\begin{align*}
\frac{\widetilde{\boldsymbol{T}}^T\widetilde{\boldsymbol{T}}}{n} \overset{p}{\to} \begin{bmatrix}0.5&0\\0&0.5\end{bmatrix} = \boldsymbol{M}.
\end{align*} 
We further define
\begin{align*}
\boldsymbol{A}&= \boldsymbol{K} \boldsymbol{M}^{-1} \bigg[ \frac{\widetilde{\boldsymbol{T}}^T(\boldsymbol{Z} \boldsymbol{\gamma}+ \boldsymbol{\epsilon})}{n}\bigg],\\
\boldsymbol{B}&=\boldsymbol{K}\bigg[ \Big(\frac{\widetilde{\boldsymbol{T}}^T\widetilde{\boldsymbol{T}}}{n}\Big)^{-1} - \boldsymbol{M}^{-1}\bigg]\bigg[ \frac{\widetilde{\boldsymbol{T}}^T(\boldsymbol{Z} \boldsymbol{\gamma}+ \boldsymbol{\epsilon})}{n}\bigg],\\
\end{align*}
so that $\hat{\tau}_{\textup{PSR}}=\boldsymbol{A}+\boldsymbol{B}$.

For $\boldsymbol{A}$, with some algebra, we can show
\begin{align*}
\boldsymbol{A}&=\frac{2}{n}\bigg[\sum_{j=1}^{p}\sum_{i=1}^{n}(2T_i-1)\gamma_j z_{i,j} + \sum_{i=1}^{n}(2T_i-1) \epsilon_{i}\bigg].
\end{align*}

For the first term on the right, we have 
\begin{align*}
\sum_{j=1}^{p}\sum_{i=1}^{n}(2T_i-1)\gamma_j z_{i,j} = \sum_{j=1}^{p} \gamma_j \bigg[\sum_{i \in \{i:T_i=1\}}z_{i,j} - \sum_{i \in \{i:T_i=0\}}z_{i,j}\bigg].
\end{align*}
where $\{i:T_i=1\}$ and $\{i:T_i=0\}$ represent the two treatment groups.  From the proof of Theorem \ref{thm:MahaDistLimitDist}, we understand that $\sum_{i \in \{i:T_i=1\}}z_{i,j} - \sum_{i \in \{i:T_i=0\}}z_{i,j}$ is a stationary process under the proposed method (i.e. a mean reverting process as $n \to \infty$).  Therefore, 
\begin{align*}
\sum_{i \in \{i:T_i=1\}}z_{i,j} - \sum_{i \in \{i:T_i=0\}}z_{i,j}&=O_p(1),\\
\sum_{j=1}^{p} \gamma_j [\sum_{i \in \{i:T_i=1\}}z_{i,j} - \sum_{i \in \{i:T_i=0\}}z_{i,j}]&=O_p(1).
\end{align*}
In addition, note that $(2T_i-1)^2=1$, we have
\begin{align*}
\textup{Var}\bigg(\frac{2}{n}\sum_{i=1}^{n}(2T_i-1) \epsilon_{i}\bigg)&=\mathbb{E}\bigg(\frac{4}{n^2}\sum_{i=1}^{n}(2T_i-1)^2 \epsilon_{i}^2\bigg)\\
&=\mathbb{E}\bigg(\frac{4}{n^2}\sum_{i=1}^{n}\epsilon_{i}^2\bigg)\\
&=\frac{4\sigma_{\epsilon}^2}{n}.
\end{align*}
Therefore, 
\begin{align*}
\sqrt{n}\boldsymbol{A} \overset{D}{\rightarrow}N(0,4\sigma^2_{\epsilon}).
\end{align*}
Similarly, for $\boldsymbol{B}$, we will show $\sqrt{n}\boldsymbol{B} \overset{p}{\to} 0$.  First note that
\begin{align*}
 \Big(\frac{\widetilde{\boldsymbol{T}}^T\widetilde{\boldsymbol{T}}}{n}\Big)^{-1} - \boldsymbol{M}^{-1} \overset{p}{\to} 0.
\end{align*}
Therefore, showing $\sqrt{n}\boldsymbol{B} \overset{p}{\to} 0$ is equivalent to show
\begin{align*}
\frac{\widetilde{\boldsymbol{T}}^T(\boldsymbol{Z} \boldsymbol{\gamma}+ \boldsymbol{\epsilon})}{\sqrt{n}}=O_p(1).
\end{align*}
First, notice that
\begin{align*}
\frac{\widetilde{\boldsymbol{T}}^T(\boldsymbol{Z} \boldsymbol{\gamma}+ \boldsymbol{\epsilon})}{\sqrt{n}}=\frac{1}{\sqrt{n}}\begin{bmatrix} \sum_{j=1}^{p}\sum_{i=1}^{n}T_i\gamma_j z_{i,j} + \sum_{i=1}^{n}T_i \epsilon_{i}\\\sum_{j=1}^{p}\sum_{i=1}^{n}(1-T_i)\gamma_j z_{i,j} + \sum_{i=1}^{n}(1-T_i) \epsilon_{i} \end{bmatrix}.
\end{align*}
Since 
\begin{align*}
\frac{1}{\sqrt{n}}\Big(\sum_{j=1}^{p}\sum_{i=1}^{n}T_i\gamma_j z_{i,j} + \sum_{i=1}^{n}T_i \epsilon_{i}\Big)=&\frac{1}{2}\bigg[\frac{1}{\sqrt{n}}\Big(\sum_{j=1}^{p}\sum_{i=1}^{n}\gamma_j z_{i,j} + \sum_{i=1}^{n} \epsilon_{i}\Big) + \\ &\frac{1}{\sqrt{n}}\Big(\sum_{j=1}^{p}\sum_{i=1}^{n}(2T_i-1)\gamma_j z_{i,j} + \sum_{i=1}^{n}(2T_i-1) \epsilon_{i}\Big)\bigg].
\end{align*}
By central limit theorem, we have
\begin{align*}
\frac{1}{\sqrt{n}}\Big(\sum_{j=1}^{p}\sum_{i=1}^{n}\gamma_j z_{i,j} + \sum_{i=1}^{n} \epsilon_{i}\Big) =O_p(1).
\end{align*}
In addition,
\begin{align*}
\frac{1}{\sqrt{n}}\Big(\sum_{j=1}^{p}\sum_{i=1}^{n}(2T_i-1)\gamma_j z_{i,j} + \sum_{i=1}^{n}(2T_i-1) \epsilon_{i}\Big)=\frac{\sqrt{n}\boldsymbol{A}}{2}.
\end{align*}
Since $\sqrt{n}\boldsymbol{A}$ converges to a normal distribution, 
\begin{align*}
\frac{1}{\sqrt{n}}\Big(\sum_{j=1}^{p}\sum_{i=1}^{n}(2T_i-1)\gamma_j z_{i,j} + \sum_{i=1}^{n}(2T_i-1) \epsilon_{i}\Big)=O_p(1).
\end{align*}
Therefore,
\begin{align*}
\frac{1}{\sqrt{n}}\Big(\sum_{j=1}^{p}\sum_{i=1}^{n}T_i\gamma_j z_{i,j} + \sum_{i=1}^{n}T_i \epsilon_{i}\Big)=O_p(1).
\end{align*}
By symmetry, we have 
\begin{align*}
\frac{1}{\sqrt{n}}\Big(\sum_{j=1}^{p}\sum_{i=1}^{n}(1-T_i)\gamma_j z_{i,j} + \sum_{i=1}^{n}(1-T_i)  \epsilon_{i}\Big)=O_p(1).
\end{align*}
Therefore,
\begin{align*}
\frac{\widetilde{\boldsymbol{T}}^T(\boldsymbol{Z} \boldsymbol{\gamma}+ \boldsymbol{\epsilon})}{\sqrt{n}}=O_p(1).
\end{align*}
Hence, $\sqrt{n}\boldsymbol{B} \overset{p}{\to} 0$, together with $\sqrt{n}\boldsymbol{A} \overset{D}{\rightarrow}N(0,4\sigma^2_{\epsilon})$, by Slutsky's theorem, we have
\begin{align*}
\sqrt{n}(\hat{\tau}_{\textup{PSR}}-(\mu_1-\mu_2)) \overset{D}{\to} 
N(0,4\sigma^2_{\epsilon}).
\end{align*}

For $\hat{\tau}_{\textup{CR}}$, $\tilde{\tau}_{\textup{PSR}}$, and 
$\tilde{\tau}_{\textup{CR}}$, we can obtain their asymptotic distributions in 
similar ways.  Please see supplementary materials for details.
\end{proof}

\bibliographystyle{apalike}
\bibliography{YichenBib}
\end{document}